\documentclass[12pt]{article}
\usepackage[utf8]{inputenc}
\usepackage{amsmath,amsfonts,amssymb,mathrsfs,amscd,amsthm}
\usepackage{cite}
\usepackage{graphicx}
\usepackage{diagbox}
\usepackage[affil-it]{authblk}

\usepackage{hyperref}
\hypersetup{
colorlinks=true,
linkcolor=red,
citecolor=blue,
urlcolor=blue
}
\usepackage{appendix}

\allowdisplaybreaks
\usepackage{multirow}

\usepackage[a4paper, margin = 1in]{geometry}

\def\R{\mathbb{R}}
\def\C{\mathbb{C}}
\newcommand{\pd}[2]{\dfrac{\partial #1}{\partial #2}}

\newcommand{\der}[2]{\dfrac{\mathrm{d} #1}{\mathrm{d} #2}}

\newcommand{\scalarproduct}[2]{\langle #1, #2\rangle}
\newcommand{\scalarproductleftright}[2]{\left\langle #1, #2\right\rangle}
\newcommand{\norm}[1]{\|#1\|}

\newtheorem{theorem}{Theorem}
\newtheorem{proposition}{Proposition}

\def\L{\mathcal{L}}     %superoperator
\def\D{\mathcal{D}}     %superoperator
\def\N{\mathcal{N}}     %superoperator
\def\M{\mathcal{M}}     %operator
\def\I{\mathbb{I}}      %identity
\def\d{\mathrm{d}}      %differential

\def\Tr{\mathrm{Tr}}

\title{\bf Control landscapes for high-fidelity generation of C-NOT and C-PHASE gates with coherent and environmental driving}
\author[1,2,*]{Alexander N. Pechen}
\author[1,2,\textdagger]{Vadim N. Petruhanov}
\author[1,2,\ddag]{Oleg~V.~Morzhin}
\author[1,2,\textsection]{Boris~O.~Volkov}
\affil[1]{\normalsize Department of Mathematical Methods for Quantum Technologies,\par
Steklov Mathematical Institute of Russian Academy of Sciences,\par
8~Gubkina str., Moscow, 119991, Russia;}
\affil[2]{University of Science and Technology MISIS,\par
6~Leninskiy prospekt, Moscow, 119991, Russia;}
\affil[*]{Corresponding author: apechen@gmail.com; \href{https://www.mathnet.ru/eng/person17991}{mathnet.ru/eng/person17991};}
\affil[$\dagger$]{vadim.petrukhanov@gmail.com; \href{https://www.mathnet.ru/eng/person176798}{mathnet.ru/eng/person176798};}
\affil[$\ddag$]{morzhin.oleg@yandex.ru; \href{https://www.mathnet.ru/eng/person30382}{mathnet.ru/eng/person30382};}
\affil[$\mathsection$]{borisvolkov1986@gmail.com; \href{https://www.mathnet.ru/eng/person94935}{mathnet.ru/eng/person94935}}

\date{}

\begin{document}
\maketitle

\begin{abstract}
High fidelity generation of two-qubit gates is important for quantum computation, since such gates are components of popular universal sets of gates. Here we consider the problem of high fidelity generation of two-qubit C-NOT and C-PHASE (with a detailed study of C-Z) gates in presence of the environment. We consider the general situation when qubits are manipulated by coherent and incoherent controls; the latter is used to induce generally time-dependent decoherence rates. For estimating efficiency of optimization methods for high fidelity generation of these gates, we study quantum control landscapes which describe the behaviour of the fidelity as a function of the controls. For this, we generate and analyze the statistical distributions of best objective values obtained by incoherent GRadient Ascent Pulse Engineering (inGRAPE) approach.  We also apply inGRAPE and stochastic zero-order method to numerically estimate minimal infidelity values. The results are different from the case of single-qubit gates and indicate a smooth trap-free behaviour of the fidelity.

\vspace{0.2cm}
{\bf Key words:} quantum control landscape; inGRAPE; coherent control; incoherent control; two-qubit gate; C-NOT; C-PHASE; open quantum system

\vspace{0.2cm}
% Mathematics Subject Classification
{\bf MSC:} 81Q93, 34H05 
%% 81Q93  	Quantum control
%% 34H05  	Control problems involving ordinary differential equations
\end{abstract}

\section{Introduction}

Quantum computing forms an actively developing direction within the general area of quantum technologies~\cite{Acin2018, Schleich2016}. High fidelity generation of two-qubit gates is crucial for building universal quantum computing devices~\cite{NielsenChuangBook, KitaevBook, Valiev_Uspekhi_2005}. Important examples of two-qubit gates are the controlled not C-NOT and the controlled phase shift C-PHASE gates. They both act on the controlled qubit depending on the state of the controlling qubit. If the controlling qubit is in the state $|0\rangle$, then the controlled qubit is not changed. If the controlling qubit is in the state $|1\rangle$, then C-NOT changes a~controlled qubit state $\alpha|0\rangle+\beta |1\rangle$ to $\alpha|1\rangle+\beta|0\rangle$, and C-PHASE with phase shift $\lambda$ to $\alpha|0\rangle+e^{i\lambda}\beta |1\rangle$. The particular case $\lambda=\pi$ corresponds to the Controlled Z (C-Z) gate. While C-NOT gate belongs to Clifford group which can be efficiently simulated classically according to the Gottesman--Knill theorem~\cite{Gottesman_1998}, it is a component of a~well known universal set of gates ($H$, $T$, C-NOT), which also includes one-qubit $T$ (or $\pi/8$) and Hadamard~$H$ gates and which is sufficient for implementation of arbitrary quantum schemes~\cite{KitaevBook}. This motivates the importance of high fidelity generation of gates from a~universal set. Control landscapes for high-fidelity generation of $H$ and $T$ gates were analyzed in~\cite{PetruhanovPhotonics2023-2}. In this work, we perform the analysis of the control landscapes for high-fidelity generation of the C-NOT and C-PHASE gates, thereby completing the landscape analysis for high fidelity generation of gates from a universal set for qubits controlled by both coherent control and an environmental drive.  Such analysis of quantum control landscapes is crucial for estimating efficiency of practical optimization.

Using the environment as a control is important in our study. Qubits in real experimental situations are interacting with their environments, so that they are open quantum systems. The environment is often considered as having deleterious effects on the ability to manipulate qubits. The irremovable at the moment in experiments decoherence of qubit systems results in destruction of quantum properties necessary for fast computations, and the corresponding decrease of quantum gate fidelity is a~major obstacle for experimental creation of quantum computers. However, in some cases the environment can be exploited as a~useful resource, as for example for quantum computing with mixed stated~\cite{Aharonov1999, Tarasov_2002, CiracZoller_PhysRevLett_1995}. Another example is the use of quantum Zeno or inverse Zeno effect (Zeno control)~\cite{Facchi_Gorini_Marmo_Pascazio_Sudarshan_2000,Facchi_Nakazato_Pascazio_2001,Mancini_Bonifacio_2001,Tasaki_Tokuse_Facchi_Pascazio_2004}, or non-selective quantum von Neumann measurements~\cite{Pechen_Ilin_Shuang_Rabitz_2006}, or weak measurements, e.g., for protecting entanglement in finite temperature environment~\cite{Harraz_Cong_Nieto_2021,Hou_Shi_Wang_2023}, where the measurement apparatus is used as a control. This circumstance makes important developing and analyzing control methods for optimal generation of two-qubit gates in open quantum systems, including for cases when the environment is used as a control, as we consider below. 

Methods of optimal control of quantum systems currently attract a lot of attention motivated by various applications in quantum technologies~\cite{KochEPJQuantumTechnol2022,TannorBook2007,LetokhovBook2007, FradkovBook2007,BrifNewJPhys2010,Mancini_Manko_Wiseman_2005,DongPetersen2010,Shapiro_Brumer_Book_2ndEd_2012, GoughPhilTransRSocA2012,CongBook2014, DongWuYuanLiTarn2015, KochJPhysCondensMatter2016, AlessandroBook2021, KuprovBook2023, Giannelli_Sgroi_et_al_2022}. As it was shown in~\cite{Bondar2020}, there is no single algorithm which for any given set of controls and any pair of initial and target states answers whether the initial state can be transferred into the target state or not using these controls. Despite of this negative result, for a particular class of problems such an algorithm may exist. For numerical optimization various local and global search methods are used including methods based on the Pontryagin maximum principle~\cite{BoscainPRXQuantum2021}, the GRadient Ascent Pulse Engineering (GRAPE)~\cite{Khaneja_JMagnReson_2005, SchulteHerbruggenSporlKhanejaGlaser2011,deFouquieres2011, Lucarelli_2018,PetruhanovPechenJPA2023, Goodwin_Vinding_2023}, gradient flows~\cite{Glaser2010}, Krotov-type methods~\cite{Tannor1992, BaturinaMorzhinAiT2011, Jager2014, Goerz_NJP_2014_2021, Goerz_2021, BasilewitschNewJPhys2019, Morzhin_Pechen_UMN2019, FonsecaFanchiniLimaCastelano2022}, Hamilton--Jacobi--Bellman equations~\cite{Gough_Belavkin_Smolyanov_2005}, Chopped RAndom-Basis quantum optimization (CRAB)~\cite{CanevaPRA2011, MullerSaidJelezkoCalarcoMontangero2022}, Hessian-based methods such as the Broyden--Fletcher--Goldfarb--Shanno (BFGS) algorithm~\cite{EitanPRA2011, deFouquieres2011}, genetic algorithms~\cite{Judson1992}, Maday--Turinici~\cite{Maday_Turinici_2003}, quantum feedback  control~\cite{Wiseman_Milburn_1993, Doherty_Habib_Jacobs_Mabuchi_Tan_2000, Lloyd_Viola_2001, VanHandel_Stockton_Mabuchi_2005,Mancini_Wiseman_2007,Gough_2012, Schirmer_Jonckheere_Langbein_2018}, machine learning~\cite{DongIEEE2008, Niu_Boixo_Smelyanskiy_Neven_2018,DalgaardPRA2022, ShindiIEEE2022, Giannelli_Sgroi_et_al_2022}, speed gradient approach \cite{AnanevskiiFradkov_AiT_2005, Pechen_Borisenok_IFAC2015}, quantum reinforcement~\cite{DongIEEE2008} and quantum machine learning~\cite{Biamonte_Wittek_Pancotti_Rebentrost_Wiebe_Lloyd_2017}, the dual annealing algorithm (DAA)~\cite{Morzhin_Pechen_LJM_2021}, Lyapunov control~\cite{Khari_Rahmani_Daeichian_Mehri-Dehnavi_2022}, combined approach via the quantum optimal control suite (QuOCS)~\cite{QuOCS2022}, etc. Many methods were developed for two-qubit systems. For example, geometric control theory was used for optimal manipulation by two spin-1/2 particles~\cite{Sugny_2010}, feedback was exploited for two-qubit systems in a dissipative environment~\cite{Rafiee_Nourmandipour_Mancini_2016,Rafiee_Nourmandipour_Mancini_2017}, quasi-Newton algorithm was applied for developing a method for finding optimal two-qubit gates in recurrence protocols of entanglement purification~\cite{Preti_2022}, classical driving field was used for control of entanglement in two-qubit systems~\cite{Mojaveri_Dehghani_Taghipour_2022}, Pontryagin maximum principle was applied for time optimal realization of two-qubit entangler~\cite{Jafarizadeh_Naghdi_Bazrafkan_2022}, etc. Gradient based method GRAPE was originally developed for engineering NMR pulse sequences~\cite{Khaneja_JMagnReson_2005}. Recently, GRAPE~\cite{PetruhanovPechenJPA2023} and gradient projection methods~\cite{MorzhinPechenQIP2023} were extended to coherent and incoherent control of open quantum systems.

The efficiency of local search methods depends on the existence or absence of points of local but not global extrema (maxima for the problem of maximization of the objective, or minima for the problem of minimization of the objective) in quantum control landscapes. The analysis of extrema of quantum objective functionals is important, because such points, if they would exist, would inhibit from finding globally optimal controls using local search algorithms which, otherwise, could be very effective. For this reason points of local but not global extrema of an objective functional are called {\it traps}. High importance of the analysis of quantum control landscape was posed in~\cite{RabitzHsiehRosenthal2004}. After that, various results for closed and open quantum systems were obtained~\cite{HoRabitz2006, MoorePRA2011, Pechen2011, Pechen2012,IJC2012, FouquieresSchirmer,Larocca2018, Zhdanov2018, Russell2018,VolkovMorzhinPechenJPA2021, DalgaardPRA2022}. For controlled open quantum systems, study of the control landscape in the so called kinematic picture was performed by representing completely positive trace-preserving dynamics as points of a complex Stiefel manifold (strictly speaking, of some quotient space of a complex Stiefel manifold) for general $N$-level open quantum systems~\cite{OzaJPA2009}. Absence of traps was proved rigorously for a single qubit~\cite{Pechen2012}. Trapping behaviour was discovered for some quantum systems with special symmetries and with the dimension greater than two~\cite{Pechen2011,FouquieresSchirmer, VolkovPechenUMN2023, ElovenkovaQuantumReport2023}. 
Research on the structure of the landscape in the vicinity of a~critical point for the problem of one-qubit phase shift type gate generation was continued on a~small time scale. Absence of traps for the one-qubit gate generation problem was shown for a~rather large time in~\cite{Pechen2012} and for a~small time in~\cite{VolkovMorzhinPechenJPA2021}, where the unique singular critical point  was studied and shown to be a~saddle point. Numbers of positive and negative eigenvalues of the Hessian of the objective for some other examples of quantum systems were calculated in~\cite{MoorePRA2011}. The kinematic control landscape for transition probability in a~three-level quantum system with dynamical symmetry driven by coherent control assisted  by measurement-induced incoherent control was completely characterized~\cite{ElovenkovaQuantumReport2023}.

In this work, we involve a general approach for the analysis of quantum control landscapes based on the statistical distribution of best objective values obtained by a local search algorithm.  Using this approach, we perform the analysis of control landscapes for generation of the two-qubit quantum C-NOT and C-PHASE gates (for some $\lambda$, with a most detailed study of the C-Z gate) for three models of two-qubit quantum systems interacting with an environment when both Hamiltonian and dissipative aspects of the dynamics are controlled. The dynamics of the two-qubit system interacting with an environment is defined by a master equation with Gorini--Kossakowski--Sudarshan--Lindblad (GKSL) type's generator derived in the weak coupling limit. Coherent control enters in the Hamiltonian part of the dynamics and incoherent control in the dissipative part. The latter makes the decoherence rates generally time-dependent. To reveal whether symmetry between qubits affects the control landscape or not, three types of the Hamiltonian are considered, two of them in the absence of controls have assymetry between qubits, while one case has symmetric free Hamiltonian. The symmetric system is fully controllable in the absence of the environment while asymmetric systems are not. Two essentially different kinds of objective functionals for the problem of generation of C-NOT and C-PHASE gates are considered,  (1) defined by the sum of squared distances for the full basis of Hermitian $4\times 4$ matrices, and (2) defined by the sum of squared distances or its linear part for the three special states as proposed in~\cite{Goerz_NJP_2014_2021, Goerz_2021}. First, expressions for gradient and Hessian of these objective functionals have been derived using general approach. Then, an analysis of minimal reachable values of the objective functionals for various values of the parameters and intensity of the interaction with the environment has been made using gradient-free global search method DAA of stochastic optimization which is available in the Python library {\tt SciPy} (for its implementation see~\cite{dual_annealing_SciPy} and for discussion~\cite{Tsallis1988, TsallisStariolo1996, Xiang1997, Xiang2000}). To essentially speed up the computation for piecewise constant controls we use products of matrix exponents instead of solving differential evolution equations. After that, the quantum control landscapes for generation of the C-NOT and C-PHASE gates with various phases were studied using gradient method in a finite-dimensional space of controls. The gradient descent method (inGRAPE) has been applied for optimization of the two objective functionals taking into account non-negativity of the incoherent control. An influence of the magnitude of the decoherence rate on optimal values and the number of iterations of inGRAPE is studied. Histograms showing distributions of the optimized values of the functionals for the C-NOT gate and the C-PHASE gates with different phases are constructed. For the C-NOT and C-Z gates we also numerically obtain the dependence of minimal infidelities for various intensities of interaction of the qubits with the environment. 

For single-qubit $H$ and $T$ gates, our previous study~\cite{PetruhanovPhotonics2023-2} revealed a smooth distribution with only one peak for best obtained by inGRAPE values of fidelity for generation of $H$ gate (for all considered objectives), while the distribution of the optimized fidelity values  for generation of $T$ gate was found to have two distant peaks with the corresponding controls clearly and distantly separated in the control space. Importantly, in contrast to this case of single qubit $T$ gate generation and similarly to $H$ gate generation, here we obtain smooth single-peak distributions of best obtained by inGRAPE fidelities indicating a smooth  trap-free behaviour. In the two-qubit case studied in this work we find that only a strongly restricted control space leads to the appearance of a small second peak. Interesting that both $H$ and C-NOT with smooth single-peak distributions of best obtained by inGRAPE fidelity values are Clifford, while $T$ is not. Whether it is a coincidence or there is some underlying general reason, is an open question. Beyond two-qubit gates, three-qubit gates such as Toffoli gate, Fredkin gate, holonomic quantum gates~\cite{Maity_Purkait_2020} are also important for quantum computations and can be implemented.

The rest of our paper is organized as follows. In~Sec.~\ref{Section2}, we overview the general idea of the incoherent control method. Sec.~\ref{Section3} formulates various approaches to the problem of generating unitary gates under dissipative dynamics depending on coherent and incoherent controls. In Sec.~\ref{Sec:ME}, master equation for the two-qubit system and three different classes of Hamiltonian, with free dynamics either symmetric or non-symmetric with respect to interchange of the qubits, are provided. In~Sec.~\ref{Section5}, we adapt for the problem of two-qubit gate generation the gradient-based (inGRAPE) and zeroth-order optimization approach for piecewise constant controls. Sec.~\ref{Section6} describes the obtained numerical results. Appendices~A,~B,~C contain, correspondingly, short outline of the approach of M.Y. Goerz, D.M. Reich, and C.P. Koch to generation of unitary gates, realification of the dynamics and the objective functionals, and derivation of gradient and Hessian of the  objective functionals in the functional space of controls, both for general case and for the considered particular objectives.  Discussion  Sec.~\ref{Section_Conclusion} resumes the~work.

\section{Incoherent Control by the Environment}
\label{Section2}

Coherent control of quantum systems is realized by lasers (coherent radiation) or electromagnetic field and is deeply investigated~\cite{BrifNewJPhys2010, KochEPJQuantumTechnol2022, TannorBook2007, AlessandroBook2021, KuprovBook2023}. Coherent control mainly appears in the Hamiltonian part of the dynamics, but if the system interacts with the environment then this interaction can also transfer coherent control from the Hamiltonian part to the dissipative part~\cite{DannTobalinaKosloffPRA2020}. While it can be important in various situations, here we neglect this dependence. Incoherent control in opposite, basically appears in the dissipative part, while also can modify the Hamiltonian via Lamb shift. Its main effect is in affecting the relaxation rates between different pairs of states making them controlled and even generally time-dependent, so that the system dynamics is described by the master equation (in the system of units where the reduced Planck's constant~$\hbar$ equals~1) for the system density matrix $\rho$ ($\rho^\dagger = \rho \geq 0$, $\mathrm{Tr}\rho = 1$) 
\begin{equation}\label{Eq:ME2}
\der{\rho(t)}{t} = \L^{u, n}_t\rho(t):= -i [H^{u, n}_t, \rho(t)] + \varepsilon\underbrace{\sum_k\gamma_k(t) \D_k\rho(t)}_{\displaystyle\D^n_t\rho(t)},\quad \rho(0) = \rho_0.
\end{equation}
where the total Hamiltonian $H^{u, n}_t$ includes free Hamiltonian, interaction with coherent control $u(t)$, and possibly Lamb shift depending on incoherent control $n(t)$,  $\gamma_k(t)$ are generally time-dependent decoherence rates depending on incoherent control $n(t)$, $\D_k$ are some dissipative superoperators, $\D^{n}_t$ is the total dissipator depending on incoherent control, and the parameter $\varepsilon > 0$ describes strength of the coupling between the system and the environment. The notation $[A, B] = AB - BA$ denotes commutator of operators $A$ and $B$. Such approach with incoherently driven relaxation rates was proposed and studied in~\cite{PechenRabitz2006}, where two physical cases of $\D_k$ and $\gamma_k(t)$ corresponding to the well known in the theory of open quantum systems weak coupling and low density limits~\cite{SpohnLebowitzReview} with GKSL form were studied. Master equations derived beyond secular approximation as well as for ultrastrong-coupling and the strong-decoherence limits~\cite{Trushechkin_2021,Trushechkin_2022} may also be of interest for an investigation.

The weak coupling case has been most studied from the point of view of control since that. A~natural example of the environment for this case is the environment formed by incoherent (not necessarily thermally and even not necessarily constantly in time distributed) photons. For this case, the decoherence rates for transition between levels $i$ and $j$ with transition frequency $\omega_{ij}$ were considered in~\cite{PechenRabitz2006} as
\begin{equation*}
    \gamma_{ij}(t)=\pi\int \d {\bf k}\,\delta(\omega_{ij}-\omega_{\bf k})|g({\bf k})|^2(n_{\omega_{ij}}(t)+\kappa_{ij}),\qquad i,j = 1, \ldots, N,
\end{equation*}
where $\kappa_{ij}=1$ for $i>j$ and $\kappa_{ij}=0$ otherwise, $\omega_{\bf k}$ is the dispersion law for the bath (e.g., $\omega=|{\bf k}|c$ for photons, where ${\bf k}$ is photon momentum and $c$ is the speed of light), and $g({\bf k})$ describes coupling of the system to $\bf k$-th mode of the bath. For $i>j$, $\kappa_{ij}=1$ corresponds to spontaneous emission and $\gamma_{ij}$ determines rate of spontaneous and induced emission between levels $i$ and $j$. For $i<j$, $\gamma_{ij}$ determines rate of induced absorption. These decoherence rates appear in~(\ref{Eq:ME2}), where $k=(i,j)$ is multi-index. 
The magnitude of the decoherence rates affects the speed of decay of the off-diagonal elements of the system density matrix and determines the values of the diagonal elements towards which the diagonal part of the system density matrix evolves. 

In the above case with the environment formed by incoherent photons, incoherent control is the set of all functions $n=\{n_{\omega_{ij}}(t)\}$. Here $n_{\omega_{ij}}(t)$ is density of particles of incoherent photons at frequency $\omega_{ij}$.  
Since each incoherent control $n_{\omega_{ij}}(t)$ by its physical meaning is density of particles of the environment, it is non-negative, $n_{\omega_{ij}}(t)\ge 0$. Thus a~suitable way is to consider incoherent control functions as elements $n_{\omega_{ij}} \in L_{\infty}([0,T]; \mathbb{R_+})$ for some final time $T>0$. Coherent control has no such bound, but for symmetry we consider coherent control function also as element $u \in L_{\infty}([0,T]; \mathbb{R})$. In this case for each $(u,n)$ there is a~unique absolute continuous solution $\rho(t)$ of the equation~(\ref{Eq:ME2}). 

Incoherent control is represented by the function $n_\omega(t)$ of $\omega$ and $t$ which, for incoherent photons, described frequency distribution of photons, which can be non-thermal and even time depending. A~natural question is how useful such control can be. It this regard, it was shown that it can be quite rich --- combining coherent and incoherent controls  can be used for approximate generation of arbitrary density matrices of generic $N$-level quantum systems~\cite{PechenPRA2011}. Thus an advantage of incoherent control is that it allows, when combined with coherent control, to approximately steer {\it any} initial density matrix to {\it any} given target density matrix~\cite{PechenPRA2011}. This property approximately realizes controllability of open quantum systems in the set of all density matrices --- the strongest possible degree of quantum state control~\cite{PechenPRA2011}. This result has several important features. (1) It was obtained with dissipators $\D_k$ which correspond to a~physical class of master equations well known as the weak coupling limit in theory of open quantum systems. (2) It was obtained for generic quantum systems --- for almost all values of the physical parameters of this class of master equations and for multi-level quantum systems of an arbitrary dimension. (3) For incoherent controls in this scheme in some simple version an explicit analytic solution (not numerical) was obtained. (4) The scheme is robust to variations of the initial state --- the optimal control steers simultaneously {\it all} initial states into the target state, thereby physically realizing all-to-one Kraus maps previously theoretically exploited for quantum control~\cite{Wu_2007_5681}. While theoretically all-to-one Kraus maps were introduced a~while ago, the~first their experimental realization was done only recently for an open single qubit~\cite{Zhang_Saripalli_Leamer_Glasser_Bondar_2022}. Proposed in~\cite{PechenRabitz2006} general incoherent control scheme was further developed and patented for some applications, e.g, for interactively controlling multi-species atomic and molecular systems with $\rm Gd_2O_2S\!:\!Er^{3+}$ (6\%) samples~\cite{LaforgeJCP2018}. Beyond that, incoherent photons can play crucial role in other processes. For example, incoherent excitations were used for modelling  photosynthesis and photoisomerization in light-harvesting systems, leading to a molecular response with the time-independent steady state~\cite{Pachon_Botero_Brumer_2017,Brumer_2018}, as well as in quantum feedback-like mechanism for explaining charge separation in quantum photosinthesis~\cite{Kozyrev_Pechen_2022}. 

\section{Generation of Quantum Channels}
\label{Section3}
In this section, we formulate the problem of generation of unitary quantum gates under dissipative evolution driven by coherent and incoherent controls.

Denote the space of complex matrices as $\M_k = \C^{k\times k}$. A~linear map $\Phi: \M_N \to \M_N$ is called positive, if it preserves positivity of a matrix, $\Phi \rho \ge 0$ for any $\rho \ge 0$. A~quantum channel of an $N$-level quantum system is a~completely positive trace-preserving map, i.e., map $\Phi: \M_N \to \M_N$~\cite{Kraus_1983, Holevo_Book_2019} such that:
\begin{itemize}
    \item $\Phi$ is linear, i.e. $\Phi(\alpha\rho+\beta\sigma)=\alpha\Phi \rho+\beta\Phi\sigma$ for any $\alpha,\beta\in\mathbb C$, $\rho,\sigma\in\M_N$;
    \item trace-preserving, i.e. $\Tr(\Phi \rho) = \Tr \rho$ for any $\rho\in\M_N$;
    \item completely positive, i.e. for any $k\in\mathbb N$ the map $\Phi\otimes {\rm Id}_k$ is positive, where ${\rm Id}_k$ is identity map on the linear space $\M_k$.
\end{itemize}
Maps acting in the space of operators (matrices) are also called as superoperators. Quantum channels are also connected to non-commutative operator graphs~\cite{Amosov_2020}.

We operate with both density matrices and unitary gates (matrices) to measure how effectively does the dissipative evolution with the two types of controls approximates a~target unitary gate. For the problem of gate generation, we use the dynamical equation for the system evolution operator $\Phi^{u,n}_t$ which directly follows from~(\ref{Eq:ME2}):
\begin{equation}
\der{\Phi}{t} = \L^{u, n}_t \Phi,\qquad \Phi^{u,n}_0 = {\rm Id}_N,\qquad t \in [0,T].
\label{Eq:ME2_evolution_operator}
\end{equation}
It defines evolution of the density matrix $\rho(t) = \Phi^{u,n}_t \rho_0$.

Any unitary gate  $U$ determines a~quantum channel~$\Phi_U$ which acts on a~density matrix~$\rho$ as $\Phi_U\!\rho=U\rho U^\dagger$. There are various suggestions  how to formulate the objective functional measuring distance between the final dissipative evolution and a~target unitary gate. We formulate the problem of generation of $\Phi_U$ as  optimization of the following Mayer-type objective functional with a~fixed time~$T$:
\begin{equation}
    F_U(u, n) = J_U(\Phi^{u,n}_T) \to \inf_{u, n},
    \label{optimization_problem}
\end{equation}
where $J_U$ is a~functional on the set of quantum channels which satisfies
\begin{equation}
    J_U(\Phi) \geq 0,\qquad J_U(\Phi) = 0 \Leftrightarrow \Phi =\Phi_U,\qquad 
    \max_\Phi J_U(\Phi) = 1.
    \label{kinematic_functional_bounds}
\end{equation}
The kinematic objective functional $J_U$ is a~measure of closeness of the quantum channel $\Phi$ to the unitary quantum channel $\Phi_U$. We consider three types of $J_U(\Phi)$. 

First is defined by the squared Hilbert--Schmidt distance between the desired unitary quantum channel~$\Phi_U$ and a~quantum channel $\Phi$ (``sd'' is for squared distance):
\[
J_U^{\rm sd}(\Phi) = \frac1{2N^2}\|\Phi - \Phi_U\|^2,\qquad \|\Phi\|^2 = \Tr (\Phi^\dagger \Phi).
\] 

Second is defined by the mean value of the Hilbert--Schmidt distance between the actions of $\Phi$ and $\Phi_U$ on the set of three special density matrices $\{\rho_m\}_{m = 1}^3$:
\begin{equation*}
J_U^{\rm GRK,sd}(\Phi) = \frac{1}{6}\sum_{m = 1}^3 \|\Phi\rho_m -\Phi_U\rho_m\|^2.\label{second_functional}
\end{equation*}

Third is defined by the mean value of the Hilbert--Schmidt scalar product of the actions of $\Phi$ and $\Phi_U$ on the set of three special density matrices $\{\rho_m\}_{m = 1}^3$ (``sp'' is for scalar product):
\[
J_U^{\rm GRK,sp}(\Phi) = 1 - \frac{1}{3}\sum_{m = 1}^3 \dfrac{\Tr (\Phi\rho_m \Phi_U\rho_m)}{\Tr\rho_m^2}.
\] 
This functional is used for simplification as it is defined by the linear part of~$J_U^{\rm GRK,sd}$. 

The objective functionals~$J_U^{\rm GRK,sp}$~and~$J_U^{\rm GRK,sd}$ are defined for some specific set of three special density matrices $\{\rho_m\}_{k = 1}^3$ which was considered by M.Y.~Goerz, D.M.~Reich, and C.P.~Koch in~\cite{Goerz_NJP_2014_2021, Goerz_2021}, where it was shown for any unitary $U$ that if some quantum channel $\Phi$ satisfies $\Phi \rho_m=U\rho_mU^\dagger$ for $m = 1,2,3$, then $\Phi=\Phi_U$. Therefore achieving zero value for the objective functional~$J_U^{\rm GRK,sd}$ implies successfull generation of the target unitary gate.
Unlike the functionals~$J_U^{\rm sd}$ and~$J_U^{\rm GRK,sp}$, the functional $J_U^{\rm GRK,sp}$ does not satisfy conditions~(\ref{kinematic_functional_bounds}). For a~given set $\{\rho_m\}_{k = 1}^3$ this functional satisfies the conditions $J_U(\Phi) \geq 0$ and $J_U(\Phi) = 0 \Leftrightarrow \Phi =\Phi_U$ only on the set of quantum channels $\Phi$ which does not increase purity of the states $\rho_m$, i.e., such that $\Tr(\Phi\rho_m)^2 \le \Tr\rho_m^2$ for $m = 1,2,3$. 
For brevity we call these objective functionals as GRK-sd and GRK-sp. A~more detailed discussion of these objective functionals is provided in Appendix~\ref{Sec:AppendixA}.

According to the introduced above kinematic functionals, we specify the dynamical objective functional~$F_U(u,n)$ with a~fixed~$T$ as the following three functionals:
\[
F_U^{\rm sd}(u, n) = J_U^{\rm sd}(\Phi^{u,n}_T), \quad 
F_U^{\rm GRK,sd}(u, n) = J_U^{\rm GRK,sd}(\Phi^{u,n}_T), \quad F_U^{\rm GRK,sp}(u, n) = J_U^{\rm GRK,sp}(\Phi^{u,n}_T).
\]
Using these objective functionals, we consider two-qubit systems ($N=2$) and as target unitary gates $\text{C-NOT}$ and $\text{C-PHASE}$ gates which are defined in the computational basis by the unitary matrices
\[
\textrm{C-NOT} = \begin{pmatrix}
1 & 0 & 0 & 0\\
0 & 1 & 0 & 0\\
0 & 0& 0 & 1\\
0 & 0 & 1 & 0
\end{pmatrix},\quad \textrm{C-PHASE} (\lambda) = \begin{pmatrix}
1 & 0 & 0 & 0\\
0 & 1 & 0 & 0\\
0 & 0& 1 & 0\\
0 & 0 & 0 & e^{i\lambda}
\end{pmatrix}.
\]

\section{The Master Equation}
\label{Sec:ME}

Following~\cite{PechenRabitz2006,Morzhin_Pechen_LJM_2021}, consider the open quantum system consisting of two qubits driven by coherent $u$ and incoherent $n$ controls. We denote by $\mathbb{I}_2$ the $2\times 2$ identity matrix, and denote Pauli matrices and lowering and uppering one-qubit matrices as
\[
\sigma_x = \begin{pmatrix}
0 & 1 \\
1 & 0
\end{pmatrix},\quad
\sigma_y = \begin{pmatrix}
0 & -i \\
i & 0
\end{pmatrix},\quad
\sigma_z = \begin{pmatrix}
1 & 0 \\
0 & -1
\end{pmatrix},\quad \sigma^+ = \begin{pmatrix}
0 & 0 \\ 1 & 0
\end{pmatrix}, \quad
\sigma^- = \begin{pmatrix}
0 & 1 \\ 0 & 0
\end{pmatrix}.
\]
Master equation describing dynamics of the density matrix $\rho \in \C^{4\times4}$ has the form~(\ref{Eq:ME2}) with the superoperator $\L_t^{u, n} = -i [H_t^{u,n} , \,\cdot\, ] + \D^n_t$, where for the total Hamiltonian $H^{u,n}_t$ of the two qubits we consider three different types with the following general form ($k = 1,2,3$) 
\begin{equation}
H^{u,n}_{t,k}= H_{S,k} + \varepsilon H^n_{{\rm eff}, t} + V_k u(t).
\label{total_Hamiltonian}
\end{equation}
Here $H_{S,k}$ is the free Hamiltonian, $H^n_{\rm eff}$ is the effective Hamiltonian (Lamb shift), which depends on the incoherent control, and  $V_k$ is the Hamiltonian of interaction with coherent control. The parameter $\varepsilon > 0$ describes strength of the coupling between the system and its environment. 

System 1. Free and interaction Hamiltonians are
\begin{equation}
\begin{split}
H_{S,1} &=  \frac{\omega_1}{2} \left( \sigma_z \otimes \mathbb{I}_2 \right) 
+ \frac{\omega_2}{2} \left( \mathbb{I}_2 \otimes \sigma_z \right),\\
V_1 &= \sigma_x \otimes \mathbb{I}_2 + \mathbb{I}_2 \otimes \sigma_x. 
\end{split}
\label{system_1}
\end{equation}

System 2. Free and interaction Hamiltonians are
\begin{equation}
\begin{split}
H_{S,2} &= H_{S,1} =\frac{\omega_1}{2} \left( \sigma_z \otimes \mathbb{I}_2 \right) 
+ \frac{\omega_2}{2} \left( \mathbb{I}_2 \otimes \sigma_z \right),\\
V_2 &= \sigma_x \otimes \sigma_x.
\end{split}
\label{system_2}
\end{equation}

System 3. Free and interaction Hamiltonians are
\begin{equation}
\begin{split}
H_{S,3} &=
\sigma_z \otimes {\mathbb I}_2 + {\mathbb I}_2 \otimes \sigma_z + 
\alpha (\sigma_y \otimes \sigma_y + \sigma_z \otimes \sigma_z), \quad \alpha>0,\\
V_3 &= \sigma_x \otimes {\mathbb I}_2.
\end{split}
\label{system_3}
\end{equation} 
Systems 1 and 2 have the same free Hamiltonian with generally different frequencies of the first and the second qubit $\omega_1,\omega_2>0$.  Free Hamiltonians $H_{S,1}=H_{S,2}$  for different $\omega_1$ and $\omega_2$ are assymetric with respect to intercange of the qubits. Free Hamiltonian $H_{S,3}$ is in opposite, symmetric with respect to such interchange.  The difference between  interactions~$V_1$ and $V_2$ is that  for $V_1$ the same coherent control~$u$ addresses each qubit independently, while for $V_2$ the control~$u$ acts to couple the qubits. Such control is not most general and not most efficient, but simpler for realization while allows to generate some particular gates. The Hamiltonian $H_{S,3}+u(t)V_3$ of the third system provides a~fully unitarily controllable under coherent controls in the absence of interaction with the environment (i.e., when $\varepsilon = 0$). Systems 1 and 2 in the absence of the environment are non fully unitarily controllable, but some gates for these systems can still be generated.

The dissipative superoperator is the same for all cases. We assume that each qubit can be separately addressed by incoherent control. For systems 1 and 2, it can be done for example by using frequency dependent incoherent controls $n_{\omega_1}(t)$ and $n_{\omega_2}(t)$ at frequencies $\omega_1$ and $\omega_2$. For system 3, it can be done using either polarization dependence of incoherent control, or individual spatial addressing of the qubits. For brevity, we further denote by $n_1(t)$ and $n_2(t)$ incoherent controls acting on the first and second qubit. Then the effective Hamiltonian and the dissipative superoperator are
\begin{align}
H^n_{{\rm eff},t} &=
\Lambda_1 n_1(t) \left( \sigma_z \otimes \mathbb{I}_2 \right) + 
\Lambda_2 n_2(t) \left( \mathbb{I}_2 \otimes \sigma_z \right).\nonumber\\
\mathcal{D}^{n}(\rho) &= 
\mathcal{D}^{n_1}(\rho) + 
\mathcal{D}^{n_2}(\rho), \label{dissipator}
\\
\mathcal{D}^{n_j}_t(\rho) &=  \Omega_j (n_j(t) + 1) \left( 2 \sigma^-_j \rho \sigma^+_j - 
\sigma_j^+ \sigma_j^- \rho - \rho \sigma_j^ + \sigma_j^- \right) + \nonumber \\
&+ \Omega_j n_j(t) \left( 2\sigma^+_j \rho \sigma^-_j - 
\sigma_j^- \sigma_j^+ \rho - \rho \sigma_j^- \sigma_j^+ \right),
\qquad j = 1,2,\nonumber 
\end{align}
where $\Lambda_j>0$ and $\Omega_j>0$ are some constants depending on the details of interaction between the system and the environment, $\sigma_1^{\pm} = \sigma^{\pm} \otimes \mathbb{I}_2$, and $\sigma_2^{\pm} = \mathbb{I}_2 \otimes \sigma^{\pm}$. In this case, in~(\ref{Eq:ME2}) there are four time-dependent decoherence rates and four corresponding dissipators corresponding to transitions from the ground to the excited state ($\uparrow$) and from the excited to the ground state ($\downarrow$) of $j$th qubit ($j=1,2$),
\begin{align*}
\gamma_{j,\uparrow}&=\Omega_j n_j(t), \quad 
{\mathcal D}_{j,\uparrow}=2\sigma^+_j \rho \sigma^-_j - 
\sigma_j^- \sigma_j^+ \rho - \rho \sigma_j^- \sigma_j^+,\\
\gamma_{j,\downarrow}&=\Omega_j (n_j(t) + 1),\quad
{\mathcal D}_{j,\downarrow}=2 \sigma^-_j \rho \sigma^+_j - \sigma_j^+ \sigma_j^- \rho - \rho \sigma_j^ + \sigma_j^- .
\end{align*}
All these three dynamical systems belong to the class of bilinear homogeneous systems.

\section{Optimization Techniques}
\label{Section5}

\subsection{Piecewise-Constant Controls}

We consider  the class of piecewise constant controls at a~uniform time grid, which  provides reduction of the infinite-dimensional optimization problem to a~finite-dimensional optimization problem:
\begin{align}
u(t) = \sum_{k=1}^K \theta_{[t_k, t_{k+1})}(t) u_k, \quad n_l(t) = \sum_{k=1}^K \theta_{[t_k, t_{k+1})}(t) n_{k,l},  \quad l = 1,2, \quad t \in [0, T],
\label{piecewise_constant_controls}  
\end{align}
where time points $0 = t_0 < t_1 < \ldots < t_K = T$ are fixed, the step $\Delta t = t_{k+1} - t_k = T/K$, and $\theta_{[t_k, t_{k+1})}$ is the characteristic function of $[t_k, t_{k+1})$. For brevity denote the controls by $f = (u, n_1, n_2)$. Each such piecewise constant control is described by a~$3K$-dimensional vector with components $f_{k,\mu}$ for $\mu = 1,2,3$ and $k = 1, \ldots, K$. For DAA we consider the following constraints:
\begin{equation*}
|u(t)| \leq u_{\max}, \quad 0 \leq n_l(t) \leq n_{\max}, \quad l = 1,2, \quad u_{\max}, n_{\max} > 0, \quad t \in [0, T].
\end{equation*}

\subsection{inGRAPE Method}
\label{subsection5.1}

For the analysis of control landscapes for generation of C-NOT and C-PHASE gates, we use a~GRAPE-type local search method. We introduce incoherent version of GRAPE (inGRAPE) which adapts the known method to using both coherent and incoherent controls. Introduce the auxiliary control $w = (w_1, w_2)$ such that:
\begin{equation*}
    n_l(t) = w_l^2(t), \quad l=1,2, \quad t \in [0, T].
\end{equation*}
This takes into account that incoherent controls $n_1$ and $ n_2$ should be non-negative. Since controls $n_1$ and $n_2$ are piecewise constant~(\ref{piecewise_constant_controls}), controls $w_1$ and $w_2$ are also piecewise constant:
\begin{equation}
    w_l(t) = \sum_{k=1}^{K} \theta_{[t_k, t_{k+1})}(t) w_{k,l}, \quad l = 1,2.\label{piecewise_constant_control_w}
\end{equation}
Here $n_{k,l} = (w_{k,l})^2$. The parameters $u_k, w_{k,1}, w_{k,2}$ form the $3K$-dimensional control vector $g = (u, w_1, w_2)$ to be optimized. 

For control vector $g$, we adapt the GRAPE method for the  $3K$-dimensional optimization problem~(\ref{optimization_problem}). The  iterative formula for the $(j + 1)$th step of the algorithm is:
\begin{equation}
    g^{(j + 1)} = g^{(j)} - h_j \mathrm{grad}_g F_U(f^{(j)}), \quad j = 0, 1, \dots,     \label{2qubits_gradient_descent}
\end{equation}
where $F_U$ means $F_U^{\rm sd}$ or $F_U^{\rm GRK,sd}$ or $F_U^{\rm GRK,sp}$; $g^{(j)} = \left(u^{(j)}, w_1^{(j)}, w_2^{(j)}\right)$ is the $j$th iteration of the control $g$; $f^{(j)} = \left(u^{(j)}, {n_1^{(j)}}, {n_2^{(j)}}\right)$, where $n_i^{(j)}=\left(w_l^{(j)}\right)^2$, $l=1,2$; $h_j$ is the step size parameter at the $j$th iteration; $\mathrm{grad}_g = \left(\dfrac{\partial}{\partial g_{1,1}}, \dots , \dfrac{\partial}{\partial g_{K,3}} \right)$,  and the expression for the gradient of the objective functionals is provided in equations~(\ref{gradient_pconst_F1})--(\ref{gradient_pc_Psi_x}) from Appendix~\ref{Sec:AppendixC}. There is an iterative scheme for numerical computing this expressions for the gradients of the objective functionals which has linear time complexity $O(K)$ by number of intervals $K$. The integral parts of the expressions for the gradients are computed numerically by trapezoidal rule with $20$ segments. For example, the classical GRAPE scheme~\cite{Khaneja_JMagnReson_2005} proposes to use approximation of an integrand up to the first linear term. This makes the computation faster but less accurate comparing to approximate calculation of the integral.

The gradients of the objective functionals $F_U^{\rm GRK,sd}$ and $F_U^{\rm GRK,sp}$ are numerically computed much faster than the gradient of the objective functional $F_U^{\rm sd}$. This fact justifies using the objective functionals $F_U^{\rm GRK,sd}$ and $F_U^{\rm GRK,sp}$ instead of $F_U^{\rm sd}$ for optimization in the context of the problem of unitary gate generation.

Iterations of the gradient descent stop when the following condition is satisfied:
\begin{equation*}
\left|\mathrm{grad}_{g} F_U\left(f^{(j)}\right)\right| < \epsilon_{\rm acc},
\end{equation*}
which approximately corresponds to a~local extremum point, i.e., to a~point where the gradient equals zero according to the first order optimality condition. Depending on the functional, the accuracy parameter values are chosen as $\epsilon_{\rm acc} = 5\cdot10^{-3}, 2.5\cdot10^{-3}, 10^{-3}, 5\cdot10^{-4}, 3.125\cdot10^{-4}$.

We use an adaptive choice of the iteration step length~$h_j$, such that the step length depends on whether the functional value decreases or not, in order to provide a~strictly decreasing sequence of the objective functional values. More explicitly, we set a~starting step length $h_0$ and for every iteration, if the functional value decreases, $F_U\left(f^{(j + 1)}\right) < F_U\left(f^{(j)}\right)$, then we move to the next iteration with an increased step length: $h_{j+1} = a h_j$, $a \geqslant 1$; otherwise, if the functional value does not decrease, then we repeat the current iteration with a~decreased step: $h_{j}:= \beta h_{j}$, $0 < \beta < 1$ until the step length becomes small enough for a~successful decrease of the objective functional value. The following values of the algorithmic parameters are used: $h_0 = 1$, $a = 1.1$ and~$\beta = 0.5$.

\subsection{Zeroth-Order Stochastic Approach}
\label{subsection5.3}

We also use DAA for the search in the parallelepiped $[-u_{\max}, u_{\max}]^{K} \times [0, n_{\max}]^{2K}$. DAA belongs to zeroth-order stochastic methods (genetic algorithm, differential evolution, particle-swarm optimization, sparrow search algorithm, etc.) whose behaviour models try to find a~global minimizer of an~objective function without using its gradient. DAA combines the simulated annealing versions from~\cite{Tsallis1988, TsallisStariolo1996} together with local search strategy~\cite{Xiang1997, Xiang2000} 

We use DAA with the default settings~\cite{dual_annealing_SciPy} except of the following: 1)~the default {\tt initial\_temp} (initial artificial temperature) is increased to $2 \times 10^4$, i.e. near in four times; 2)~{\tt maxiter} is increased from the default $10^3$ to $2 \times 10^3$; 3)~{\tt maxfun} (from which number of an~objective's calls depends) is decreased from the default $10^7$ to $3 \times 10^4$; 4)~when it is noted, we set a special initial guess $f^{(0)}$ instead of an automatically generated. 

One can expect that DAA working with large $u_{\max},~n_{\max}$ may miss such a~good point ${\bf a}$ (with respect to some objective) that is in smaller subdomain. Because of the stochastic nature of DAA, one can expect that, for the same optimization problem, the results of different trials of DAA may differ significantly even with the same deterministic settings. That is why one can perform for the same optimization problem several trials of DAA and then select the lowest computed value of the objective over the trials. 

If DAA computes a~value, e.g., of $F_U^{\rm GRK,sd}$ that is close to zero then we could conclude that the minimization problem is numerically solved. However, if the global minimal value of $F_U^{\rm GRK,sd}$ differs from zero significantly, then a~problem is how to distinguish between this case and  the case when DAA stops because of, e.g., not been able of escaping of a~trap of $F_U^{\rm GRK,sd}$ (imagine that the last is possible). In this case one can adjust the parameter {\tt initial\_temp}, for which one should use ``higher values to facilitates a~wider search ... allowing dual\_annealing to escape local minima that it is trapped in''~\cite{dual_annealing_SciPy}. 

\section{Properties of the  Quantum Control Landscapes for Generation of the C-NOT and C-PHASE Gates}
\label{Section6}

To speed up the computation, we represent the evolution for piecewise constant controls as a~product of matrix exponents that takes into account the linearity of the ODEs in states and the possibility to solve the Cauchy problems in terms of matrix exponents and matrix multiplication for obtaining the objectives' values, without using such methods as Runge--Kutta 45, etc. It was implemented for stochastic optimization for general $N$-level quantum systems driven by coherent and incoherent controls with arbitrary parameters in~\cite{PechenRabitz2006}. In this work we also implement this approach in Python for numerical simulations for the inGRAPE involving, e.g., \texttt{scipy.linalg.expm} from {\tt SciPy}, which exploits the Pad\'e approximation for matrix exponents. Because the zeroth-order stochastic approach is used also for the class of piecewise constant controls, then the corresponding Python code incorporates such the extract from the aforementioned implementation of the inGRAPE method, as was already done in~\cite{PechenRabitz2006}. Parallel computations for the described below various cases were organized.

One can expect that 1) decreasing~$T$ leads to a~smaller degree of controllability of the system; 2)~increasing~$K$ leads to extending this~degree. 
Both for the inGRAPE and DAA computations we use $T=20$ and piecewise constant controls with $K=100$. 

The following values of the system's parameters are used: 
\begin{equation*}
\quad \omega_1 = 1, \quad \omega_2 = 1.1, \quad 
\Lambda_1 = \Lambda_2 = 0.5, \quad \Omega_1 =  \Omega_2 = 0.5,\quad \alpha = 0.2. 
\end{equation*}

For DAA, the following constraints are used to define the compact search domain:   
\begin{align}
\label{bounds_for_PC}
|u_k| \leq  30, \quad n_{k,l} \in [0, 10], \quad k = 1,\ldots, K, \quad l = 1,2.
\end{align}

\subsection{Minimal Infidelities Obtained With the Zeroth-Order Stochastic and inGRAPE Approaches} 
\label{subsection6.1}  

\begin{table}[ht!!]  
\centering
\caption{Minimal values of the objective functionals $F_U^{\rm sd}$, $F_U^{\rm GRK,sd}$, 
and $F_U^{\rm GRK,sp}$ computed using DAA and inGRAPE for $U = \text{C-NOT}$ and $U = \text{C-PHASE}(\lambda)$ 
with $\lambda \in \{\pi/6, ~\pi/3, ~ \pi/2, ~ 2\pi/3,~ \pi\}$ ($\text{C-PHASE}(\pi)=\text{C-Z}$) for all the three systems. The values~$T = 20$ and $K=100$. For DAA, the bounds~(\ref{bounds_for_PC}) are used. In each cell with objective values: the left column of the upper row shows the value of the objective functional obtained using the DAA and the right column of the upper row shows the difference between the value obtained using inGRAPE algorithm and the value obtained using the DAA algorithm, the second row shows the objective's value at the initial guess. The values are rounded to three digits after the~dot.}
\vspace{0.2cm}

1.~For $F_U^{\rm GRK,sd}$ \\
\footnotesize
\begin{tabular}{|c|p{20mm}|p{20mm}|p{20mm}|p{20mm}|p{21mm}|p{20mm}|}
\hline
\diagbox[width=4em]{Sys.}{\hspace{.1cm}Gate} & \text{~\quad C-NOT} & $\text{C-PHASE}\big(\frac{\pi}{6}\big)$ & $\text{C-PHASE}\big(\frac{\pi}{3}\big)$ & $\text{C-PHASE}\big(\frac{\pi}{2}\big)$ & $\text{C-PHASE}\big(\frac{2\pi}{3}\big)$ & $\text{\qquad C-Z}$ \\ 
\hline
\multirow{2}{*}{1} & 0.048~|~0.000 \newline \text{\qquad}0.109 & 0.053~|~$-0.008$ \newline \text{\qquad}0.114 & 0.058~|~0.002 \newline \text{\qquad}0.126 & 0.076~|~0.000 \newline \text{\qquad}0.140 & 0.094~|~0.001 \newline \text{\qquad}0.151 & 0.128~|~0.001 \newline \text{\qquad}0.157 \\ \hline
\multirow{2}{*}{2} & 0.071~|~$-0.011$ \newline \text{\qquad}0.152 & 0.064~|~$-0.004$ \newline \text{\qquad}0.156 & 0.067~|~$-0.007$ \newline \text{\qquad}0.167 & 0.068~|~$-0.007$ \newline \text{\qquad}0.181 & 0.069~|~$-0.009$ \newline \text{\qquad}0.194 & 0.066~|~$-0.007$ \newline \text{\qquad}0.205 \\  \hline
\multirow{2}{*}{3} & 0.096~|~$-0.034$ \newline \text{\qquad}0.177 & 0.083~|~$-0.024$ \newline \text{\qquad}0.172 & 0.089~|~$-0.028$ \newline \text{\qquad}0.172 & 0.092~|~$-0.033$ \newline \text{\qquad}0.173 & 0.087~|~$-0.026$ \newline \text{\qquad}0.175 & 0.096~|~$-0.035$ \newline \text{\qquad}0.176 \\ \hline
\end{tabular} 

\vspace{0.2cm}

\normalsize
2.~For $F_U^{\rm GRK,sp}$  \\
\footnotesize
\begin{tabular}{|c|p{20mm}|p{20mm}|p{20mm}|p{20mm}|p{21mm}|p{20mm}|}
\hline
\diagbox[width=4em]{Sys.}{\hspace{.1cm}Gate} & \text{~\quad C-NOT} & $\text{C-PHASE}\big(\frac{\pi}{6}\big)$ & $\text{C-PHASE}\big(\frac{\pi}{3}\big)$ & $\text{C-PHASE}\big(\frac{\pi}{2}\big)$ & $\text{C-PHASE}\big(\frac{2\pi}{3}\big)$ & $\text{\qquad C-Z}$ \\ 
\hline
\multirow{2}{*}{1} &  0.025~|~$-0.001$ \newline \text{\qquad}0.203 & 0.021~|~$-0.002$ \newline \text{\qquad}0.200 & 0.032~|~$-0.001$ \newline \text{\qquad}0.212 & 0.052~|~$-0.002$ \newline \text{\qquad}0.226 & 0.083~|~$-0.01$ \newline \text{\qquad}0.237 & 0.105~|~0.001 \newline \text{\qquad}0.243 \\ \hline
\multirow{2}{*}{2} &  0.088~|~$-0.005$ \newline \text{\qquad}0.227 & 0.083~|~0.001 \newline \text{\qquad}0.226 & 0.084~|~$-0.005$ \newline \text{\qquad}0.238 & 0.082~|~$-0.005$ \newline \text{\qquad}0.252 & 0.083~|~$-0.006$ \newline \text{\qquad}0.265 & 0.091~|~$-0.013$ \newline \text{\qquad}0.275 \\  \hline
\multirow{2}{*}{3} & 0.041~|~$-0.005$ \newline \text{\qquad}0.229 & 0.032~|~0.000 \newline \text{\qquad}0.212 & 0.032~|~0.002 \newline \text{\qquad}0.213 & 0.036~|~0.001 \newline \text{\qquad}0.214 & 0.037~|~0.005 \newline \text{\qquad}0.215 & 0.046~|~0.005 \newline \text{\qquad}0.217 \\ \hline
\end{tabular}  

\vspace{0.2cm}
\normalsize
3.~For $F_U^{\rm sd}$  \\
\footnotesize
\begin{tabular}{|c|p{20mm}|p{20mm}|p{20mm}|p{20mm}|p{21mm}|p{20mm}|}
\hline
\diagbox[width=4em]{Sys.}{\hspace{.1cm}Gate} & \text{~\quad C-NOT} & $\text{C-PHASE}\big(\frac{\pi}{6}\big)$ & $\text{C-PHASE}\big(\frac{\pi}{3}\big)$ & $\text{C-PHASE}\big(\frac{\pi}{2}\big)$ & $\text{C-PHASE}\big(\frac{2\pi}{3}\big)$ & $\text{\qquad C-Z}$ \\ 
\hline
\multirow{2}{*}{1} & 0.443~|~$-0.011$ \newline \text{\qquad}0.484 & 0.345~|~$-0.005$ \newline \text{\qquad}0.487 & 0.354~|~$-0.005$ \newline \text{\qquad}0.487 & 0.367~|~$-0.004$ \newline \text{\qquad}0.487 & 0.385~|~$-0.003$ \newline \text{\qquad}0.487 & 0.424~|~$-0.003$ \newline \text{\qquad}0.486 \\ \hline
\multirow{2}{*}{2} & 0.442~|~$-0.006$ \newline \text{\qquad}0.483 & 0.348~|~$-0.005$ \newline \text{\qquad}0.478 & 0.357~|~$-0.004$ \newline \text{\qquad}0.478 & 0.371~|~$-0.002$ \newline \text{\qquad}0.478 & 0.388~|~$-0.001$ \newline \text{\qquad}0.479 & 0.425~|~$-0.007$ \newline \text{\qquad}0.481 \\  \hline
\multirow{2}{*}{3} & 0.446~|~$-0.083$ \newline \text{\qquad} 0.492 & 0.425~|~$-0.056$ \newline \text{\qquad}0.491 & 0.429~|~$-0.065$ \newline \text{\qquad}0.490 & 0.429~|~$-0.065$ \newline \text{\qquad}0.489 & 0.431~|~$-0.062$ \newline \text{\qquad}0.489 & 0.439~|~$-0.074$ \newline \text{\qquad}0.489 \\ \hline
\end{tabular}   
\label{table1}
\end{table}  
 
Consider sequentially the objective functionals $F_U^{\rm sd}$, $F_U^{\rm GRK,sd}$, and $F_U^{\rm GRK,sp}$ for $U = \text{C-NOT}$ and $U = \text{C-PHASE}(\lambda)$ with $\lambda \in \{\pi/6, ~\pi/3, ~ \pi/2, ~ 2\pi/3,~ \pi\}$. Let $\varepsilon = 0.1$. 
For each such optimization problem, we carry out three trials of the DAA algorithm and one run of the inGRAPE algorithm. 

The DAA implementation in~\cite{dual_annealing_SciPy} allows, as a~variant, to set our initial guess ${\bf a}^{(0)}$. Due to this possibility, --- in order to compare the results of the DAA and inGRAPE optimization that is shown in Table~1, --- for the inGRAPE algorithm we take the initial guess  
\begin{equation}
u_k = \cos(0.3 t_k),\quad w_{k,l} = \exp\left[-5 \left(\frac{t_k}{T} - \frac{1}{2} \right)^2 \right] , \quad k=1,\ldots, K,\quad  l =1,2,\label{initial_guess_PC}
\end{equation} 
and for DAA we take $n_{k,l} = (w_{k,l})^2$ and the same $u_k$.

Because $K=100$, DAA works in the 300-dimensional search domain  taking into account the constraints~(\ref{bounds_for_PC}). 
After performing three trials of DAA, we select the minimal value among the corresponding three values of the considered objective. 
Taking into account the operation of DAA, the objectives' values obtained via DAA in Table~\ref{table1} do not pretend to be the values that are the best in principle for the considered class of controls. Nevertheless, Table~\ref{table1} shows that, for each system and gate, DAA finds visibly less values of all the considered objectives than the initial guess provides for these objectives. Table~\ref{table1} also shows for each gate and each system the difference between the best infidelity value obtained by the inGRAPE algorithm and the best infidelity value obtained by DAA. Many of these differences are negative that means that, for given initial guess, the inGRAPE algorithm provides better infidelity values than DAA.

\subsection{Analysis of the Control Landscapes for Generation of C-NOT and  C-PHASE}
\label{subsection6.2}

In this subsection, we perform a detailed study of high-fidelity control landscapes for generation of C-NOT and C-PHASE gates. We consider control landscapes for generation of C-NOT and C-Z for all three systems 1, 2, and 3, and generation of C-PHASE gate with $\lambda=\pi/2, \pi/3, \pi/4$, only for the system 3. For each case, we randomly generate $L=1000$ initial controls with a uniform distribution in some set, and then apply inGRAPE to optimize the objective starting from each of these initial conditions. inGRAPE stops when norm of the numerically computed gradient of the objective becomes smaller that the accuracy $\epsilon_{\rm acc}$. Thus we obtain for each case a distribution of optimized values of fidelity and plot histograms showing these distributions.

\begin{figure}[ht!]
    \centering
    \includegraphics[width=\linewidth]{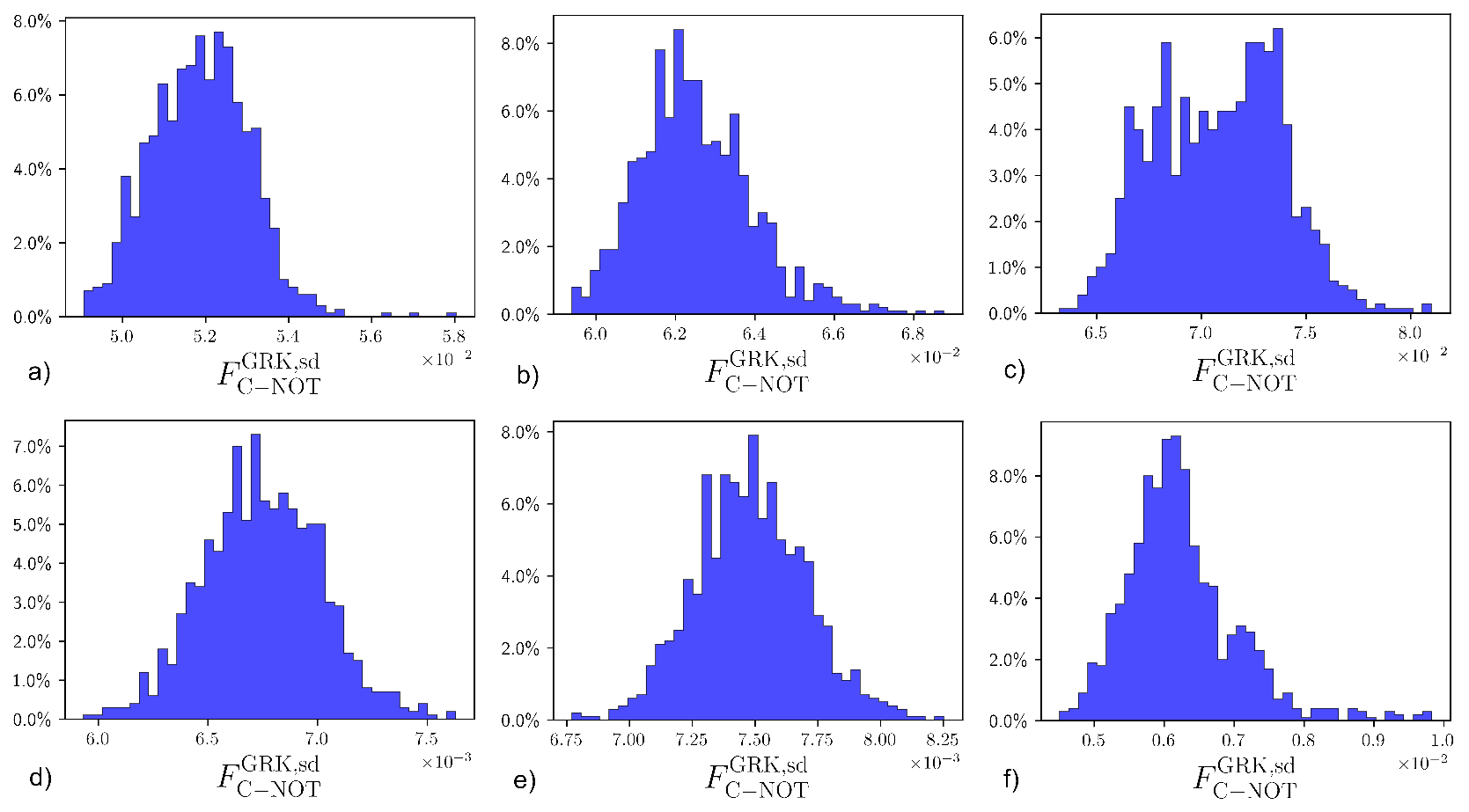}
    \caption{Histograms showing distributions of $L = 1000$ optimized by inGRAPE values of the objective functional $F^{\rm GRK, sd}_{\rm C-NOT}$ with decoherence rate $\varepsilon = 0.1$:  (a) for system~1, (b) for system~2, and (c) for system~3; with decoherence rate $\varepsilon = 0.01$: (d) for system~1, (e) for system~2, and (f) for system~3. The final time $T = 20$ and the number of time intervals $K = 100$.}
    \label{fig:Fig1}
\end{figure}

\begin{figure}[ht!]
    \centering
    \includegraphics[width=\linewidth]{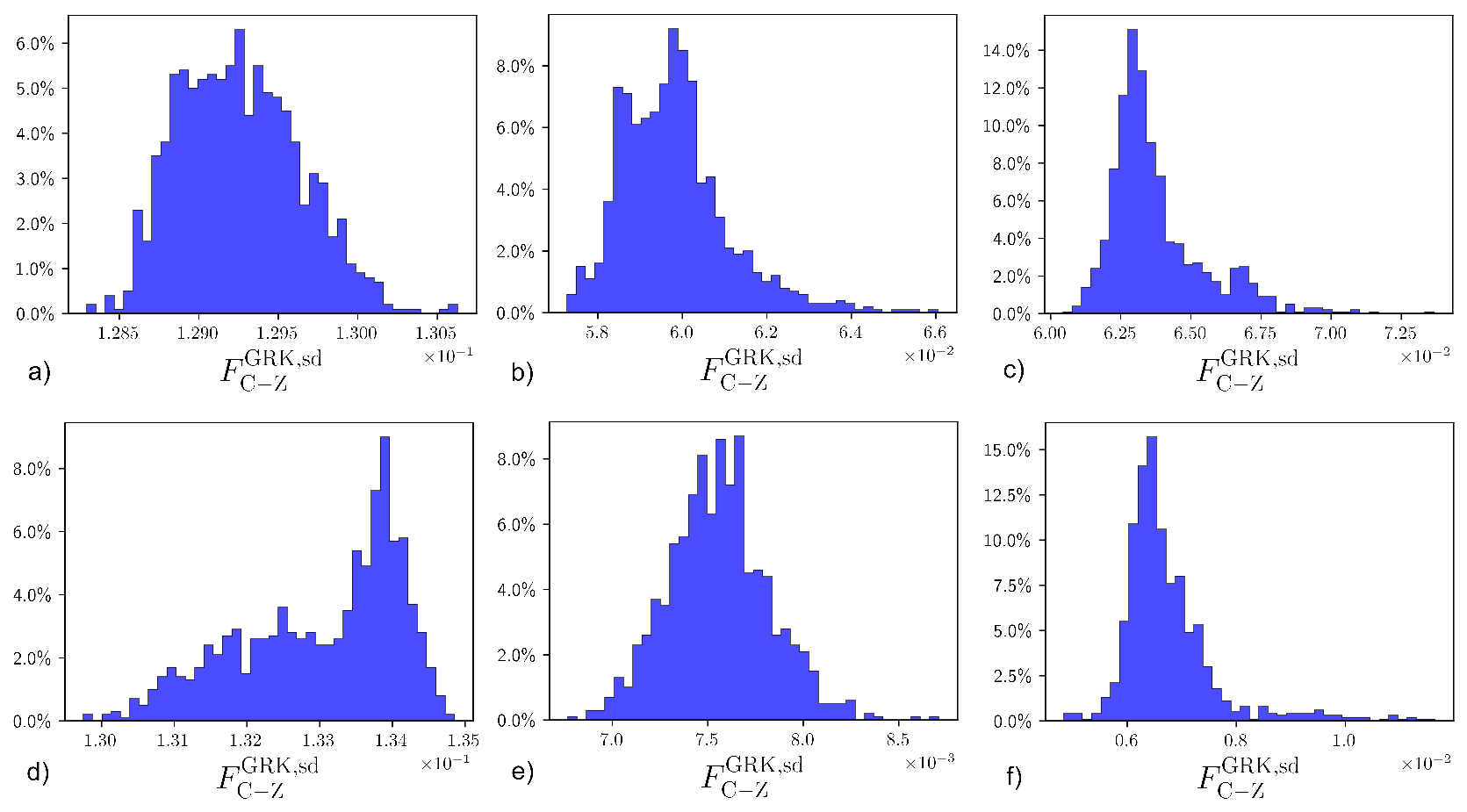}
    \caption{Histograms showing distributions of $L = 1000$ optimized by inGRAPE values of the objective functional $F^{\rm GRK, sd}_{\rm C-Z}$ with decoherence rate $\varepsilon = 0.1$:  (a) for system~1, (b) for system~2, and (c) for system~3; with decoherence rate $\varepsilon = 0.01$: (d) for system~1, (e) for system~2, and (f) for system~3. The final time $T = 20$ and the number of time intervals $K = 100$.}
    \label{fig:Fig5}
\end{figure}

\begin{figure}[ht!]
    \centering
    \includegraphics[width = \linewidth]{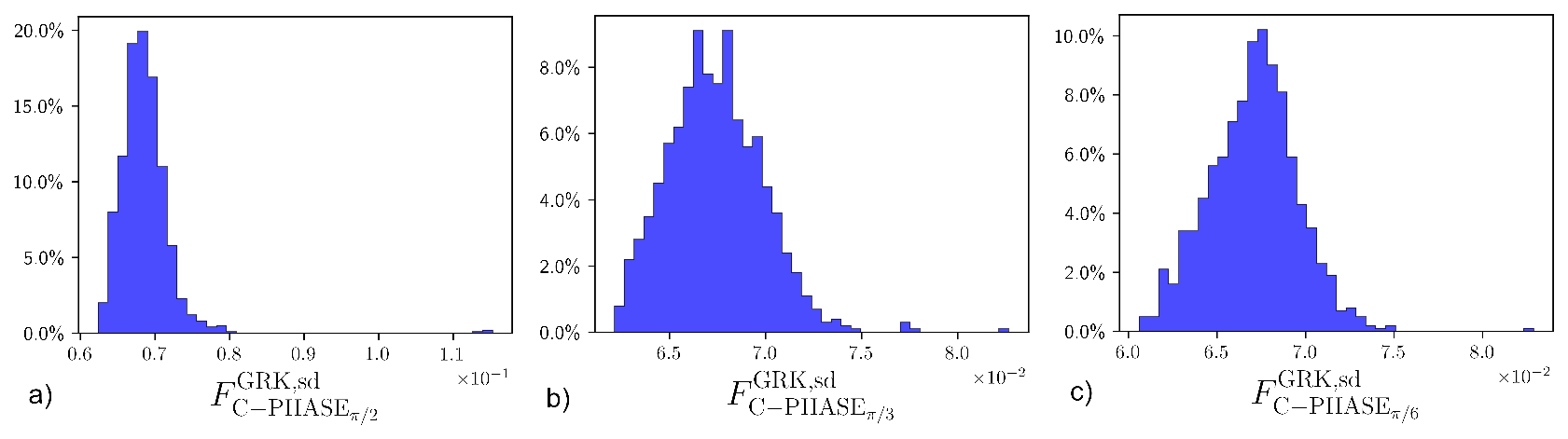}
    \caption{Histograms showing distributions of $L = 1000$ optimized by inGRAPE values of the objective functional $F^{\rm GRK, sd}_{{\rm C-PHASE}_\lambda}$ for system 3: (a)~$\lambda=\pi/2$, (b)~$\lambda=\pi/3$, and (c)~$\lambda=\pi/6$. The decoherence rate $\varepsilon = 0.1$, the final time $T = 20$, and the number of time intervals $K = 100$.}
    \label{fig:Fig2}
\end{figure}

\begin{figure}[ht!]
    \centering
    \includegraphics[width = \linewidth]{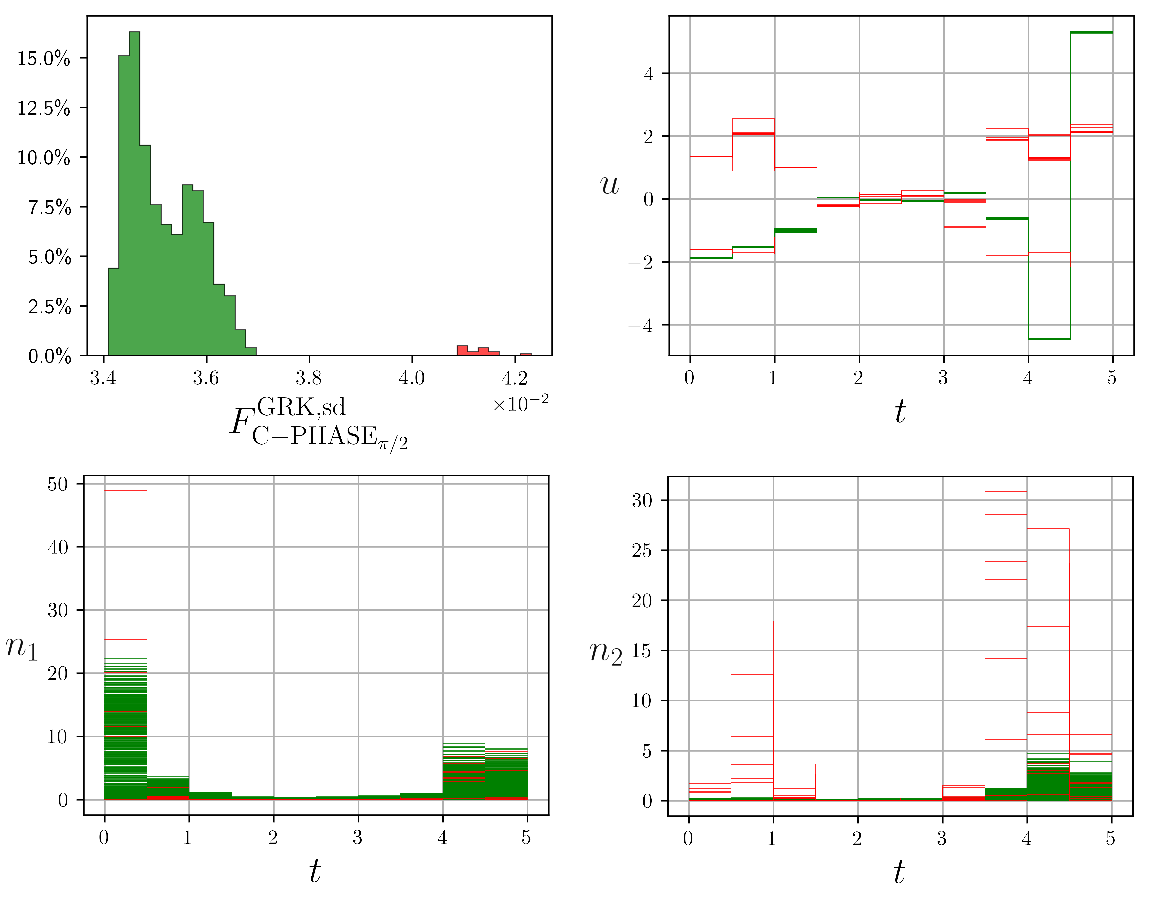}
    \caption{Results of $L = 1000$ runs of inGRAPE for minimizing the objective functional $F^{\rm GRK, sd}_{{\rm C-PHASE}_{\pi/2}}$ for the system 3 with smaller values of the final time $T = 5$ and number of time intervals $K = 10$. On subfigure (a) we plotted histograms showing the obtained distribution of the optimized values. While most values are concentrated around $3.5\times 10^{-2}$ (shown by green), there is a small fraction of values concentrated around $4.15\times 10^{-2}$ (shown by red). On subfigures~(b),~(c), and~(d) we plotted all 1000 of the optimized coherent controls and incoherent controls $n_1$ and $n_2$, respectively. The controls corresponding to green and red subgroups in subfigure (a) are separated in the space of controls.}\label{fig:twopeaks}
\end{figure}

\begin{figure}[ht!]
    \centering    \includegraphics[width=\columnwidth]{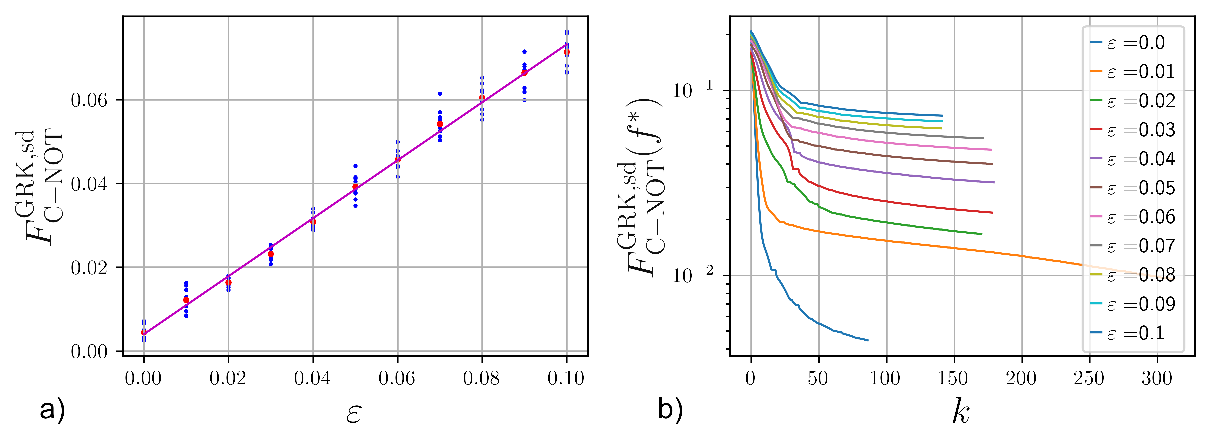}
    \caption{C-NOT gate generation for system~3. Subplot~(a): dependence of the optimized value of the objective functional $F^{\rm GRK, sd}_{\rm C-NOT}(f^{*})$ on the decoherence rate $\varepsilon$. It shows rate of decreasing the fidelity with increasing rate of decoherence.  Minimal value obtained for $\epsilon=0$ is $2.54\cdot10^{-3}$. Subplot~(b): a typical dependence of the objective functional $F^{\rm GRK, sd}_{\rm C-NOT}(f^{(j)})$ on the iteration number of the inGRAPE. Time $T = 20$, number of intervals $K = 200$, accuracy $\epsilon_{\rm acc} = 2.5\cdot10^{-3}$.}
    \label{fig:Fig4}
\end{figure}

\begin{figure}[ht!]
    \centering    \includegraphics[width=\columnwidth]{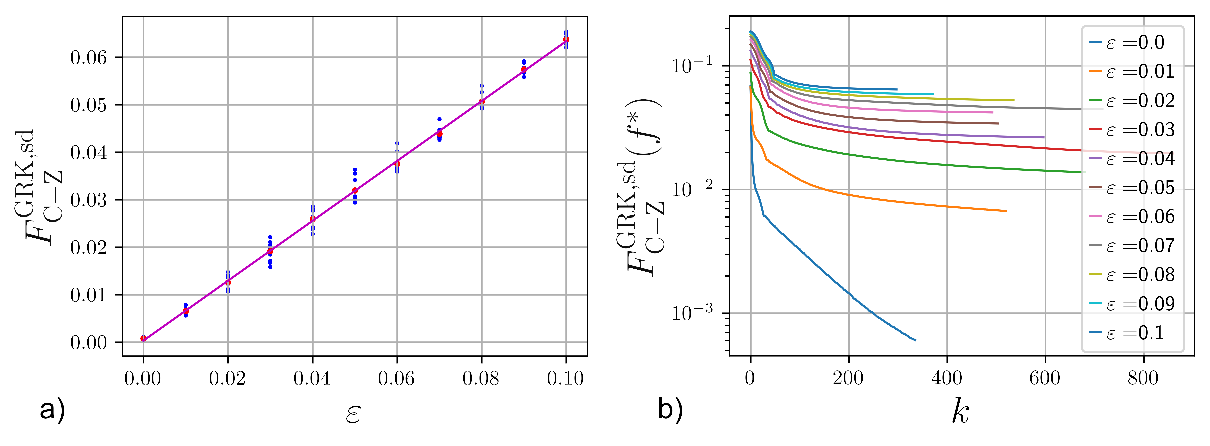}
    \caption{C-Z gate generation for system~3. Subplot~(a): dependence of the optimized value of the objective functional $F^{\rm GRK, sd}_{\rm C-Z}(f^{*})$ on the decoherence rate $\varepsilon$. It shows rate of decreasing the fidelity with increasing rate of decoherence.  Minimal value obtained for $\epsilon=0$ is $2.17\cdot10^{-4}$. Subplot~(b): a typical dependence of the objective functional $F^{\rm GRK, sd}_{\rm C-Z}(f^{(j)})$ on the iteration number of the inGRAPE. Time $T = 20$, number of intervals $K = 200$, accuracy $\epsilon_{\rm acc} = 2.5\cdot10^{-3}$.}
    \label{fig:Fig6}
\end{figure}

In Fig.~\ref{fig:Fig1} and Fig.~\ref{fig:Fig5}, we provide histograms containing results of the statistical experiments for generation of the C-NOT and C-Z gates for the system~1 [subfigure~(a)], the system~2 [subfigure~(b)], and the system 3 [subfigure~(c)]. In these experiments, we launched the gradient descent method $L = 1000$ times starting with random initial guesses $g^{(0)} = (u^{(0)}, w_1^{(0)}, w_2^{(0)})$ distributed uniformly in the hypercube $[0, 1]^{\times3K}$:
\begin{equation*}
    |u^{(0)}_{k}| \leq 1,\quad 0\leq w^{(0)}_{k,l} \leq 1, \quad l = 1,2, \quad k = 1,\dotsc,K.
\end{equation*}
Similarly, Fig.~\ref{fig:Fig2} contains  the numerical results for C-PHASE gate with $\lambda = \pi/2, \pi/3, \pi/6$ obtained only for the system~3.  

In all these figures, the obtained histograms for high-fidelity cases have smooth distributions with a single peak. Perhaps the only exception is for generation of the C-Z gate for the system~1 [subfigure~\ref{fig:Fig5}(d)], where two peaks are visible. However, in this case the obtained objective value is high and so that in this case fidelity of generation of the C-Z gate is low. It is interesting to compare these findings with control landscape analysis performed for single-qubit $H$ and $T$ gates in~\cite{PetruhanovPhotonics2023-2}. For $H$ gate, the distribution of best fidelity values was found to have on peak, whereas for $T$ gate the corresponding distribution was found to have two distantly separated peaks. Moreover, for $T$ gate optimized controls corresponding to these two peaks were found to form two groups distantly distributed in the space of controls. Our finding in this work shows that control landscapes for the considered C-NOT and C-PHASE gates for all three systems have a smooth shape with a single peak. It is similar to the corresponding distribution for the control landscape for generation of single-qubit Hadamard gate and different from the case of $T$ gate. 

Only if we strongly restrict the controls, we obtain a small second peak in the control landscape for the two-qubit gate generation. It is shown on Fig.~\ref{fig:twopeaks}, where the results of $L = 1000$ runs of inGRAPE for minimizing the objective functional $F^{\rm GRK, sd}_{{\rm C-PHASE}_{\pi/2}}$ for system 3 with a smaller value of the final time $T = 5$ and for more restricted piecewise constant controls with $K = 10$ are presented. Subfigure~(a) contains the histogram showing the obtained distribution of the optimized values. While most values are concentrated around $3.5\times 10^{-2}$ (shown by green), there is a small fraction of values concentrated around $4.15\times 10^{-2}$ (shown by red). On subfigures~(b),~(c), and~(d) we plotted all 1000 of optimized coherent controls and incoherent controls $n_1$ and $n_2$, respectively. One can see that controls corresponding to green and red subgroups in subfigure (a) are separated in the space of controls, similarly to the case of single-qubit $T$ gate as found in~\cite{PetruhanovPhotonics2023-2}.

For the objective functional $F^{\rm GRK, sd}_{\rm C-NOT}$~(\ref{second_functional}) describing C-NOT generation for the system~3, Fig.~\ref{fig:Fig4} shows dependence of (a)  optimized value of the objective functional $F^{\rm GRK, sd}_{\rm C-NOT}(f^{*})$ vs decoherence rate coefficient $\varepsilon$ and (b) a typical dependence of the objective functional value $F^{\rm GRK, sd}_{\rm C-NOT}(f^{(j)})$ on the iteration number of the gradient descent~(\ref{2qubits_gradient_descent}) for various values of the decoherence rate coefficient $\varepsilon = 0.00, 0.01, \ldots, 0.10$, with parameters: time $T = 20$, number of intervals $K = 200$, accuracy $\epsilon_{\rm acc} = 2.5\cdot10^{-3}$. As expected, the higher decoherence rate coefficient negatively impacts optimized values and, correspondingly, fidelity of the C-NOT gate generation.

\section{Discussion}
\label{Section_Conclusion}

For estimating efficiency of local search methods, an important problem is to analyze the control landscape of the underlying control problem, which describes the behaviour of the objective as a~functional of controls. In this work, we study quantum control landscapes for the problems of high fidelity generation of  two-qubit C-NOT and C-PHASE gates in open quantum systems when both coherent control and the environment are used as resources. Three types of objective functionals are considered: distance between Kraus maps as linear transformations, and two sums of distances between actual and target transformations of full basis set and of three special density matrices (that we call as GRK objective). For the study of quantum control landscapes, we develop a general approach based on the statistical distribution of best objective values obtained by a local search algorithm (inGRAPE in our case). First, we adapt gradient-free stochastic global search optimization methods to numerically estimate minimal infidelity values. Then, we derive expressions for gradient and Hessian of an arbitrary functional of a~final state and use them in the gradient-based (inGRAPE) approach to study the control landscapes of infidelities for generation of the C-NOT and C-PHASE gates. We build and analyze histograms approximating distributions of minimal objective infidelity values obtained by inGRAPE. We consider three essentially different two-qubit quantum systems, one of them is controllable and asymmetric with respect to exchange of the qubits, while two other are uncontrollable and not symmetric. For the C-NOT and C-Z gates we also numerically obtain the dependence of minimal infidelities vs the intensity of interaction between the qubits and the environment. Our findings in this work show that control landscapes for the high fidelity generation of the considered C-NOT and C-PHASE gates for all three systems have a smooth shape with a single peak. The only exception is for generation of the C-Z gate for the system~1, which has two peaks. However, in this case the obtained objective value is high and so that in this case fidelity of generation of the C-Z gate is low and does not describe the high-fidelity control landscape. It is interesting to compare these results with control landscapes for single-qubit $H$ and $T$ gates which were studied in~\cite{PetruhanovPhotonics2023-2}. For $H$ gate, the distribution of best fidelity values was found to have on peak, whereas for $T$ gate the corresponding distribution was found to have two peaks. Our findings that control landscapes for the C-NOT and C-PHASE gates have a smooth shape with a single peak is similar to the control landscape for generation of Hadamard gate and different from the control landscape for generation of $T$ gate. The underlying reason for these similarities and differences requires a further study.

\appendix
\appendixpage

\section{Goerz--Reich--Koch Approach}\label{Sec:AppendixA}
The following theorem was proved by M.Y.~Goerz, D.M.~Reich, and C.P.~Koch in~\cite{Goerz_NJP_2014_2021, Goerz_2021}.
\begin{theorem}
\label{GoerzTHM1}
Let $\rho_1=\sum_{n=1}^N \lambda_n|\phi_n\rangle \langle \phi_n|$, where $\{|\phi_n\rangle\}_{n=1}^N$ is an orthogonal basis in $\mathbb{C}^N$,  be a~density matrix with $N$ nonzero nondegenerate eigenvalues. Let a~density matrix $\rho_2$ be a~one dimensional orthogonal projector on $\mathbb{C}^N$ such that $\rho_2|\phi_n\rangle\neq 0$ for $n = 1,\ldots,N$. Let $\rho_3=\frac{1}{N}\mathbb{I}_N$. Let for some quantum channel $\Phi$ and for some unitary gate $U$ the three equalities $\Phi\rho_m=U\rho_mU^\dagger$ ($m = 1,2,3$) be satisfied. Then $\Phi \rho=U \rho U^\dagger$ for any density matrix $\rho$, i.e. $\Phi=\Phi_U$. 
\end{theorem}

Theorem~\ref{GoerzTHM1} implies that achieving zero value for the objective functional~$F_U^{\rm GRK,sd}$, implies successful generation of the target unitary gate $U$. As the dimension $N$ increases, optimization for functionals $F_U^{\rm GRK,sd}$  and $F_U^{\rm GRK,sp}$ is significantly faster than for $F_U^{\rm sd}$, because the last one depends on $N^2$ operators from $\M_N$ instead of the three operators.
But as it  is  shown in Section~\ref{Section6}, small values of the objective functionals~$F_U^{\rm GRK,sd}$  and $F_U^{\rm GRK,sp}$   do not necessarily mean small values of the objective functional $F_U^{\rm sd}$. The function  $J_U^{\rm sd}$  defines the square of the metric on the space of superoperators. It can be shown that the functionals ~$J_U^{\rm GRK,sd}$  and $J_U^{\rm GRK,sp}$ are continuous in the topology given by this metric. It implies that when the functional $F_U^{\rm sd}$  reaches sufficiently small values, the functionals~$F_U^{\rm GRK,sd}$  and $F_U^{\rm GRK,sp}$ will also take on small values.
 
In our paper, following~\cite{Goerz_NJP_2014_2021, Goerz_2021}
we consider the following three density matrices which satisfy the conditions of Theorem~\ref{GoerzTHM1} for the case of two qubits: 
\begin{align}
\label{three_initial_states}
\rho_1 = {\rm diag}\left(\frac{2}{5}, \frac{3}{10}, \frac{1}{5}, \frac{1}{10} \right), \quad \rho_2 = \frac{1}{4}J_4, \quad \rho_3=\frac{1}{4}\mathbb{I}_4.
\end{align}
Here $J_4$ means the $4 \times 4$ matrix whose all elements are equal to~1. 
Note that the action of C-NOT and C-PHASE gates does not change two density matrices of the three density matrices~(\ref{three_initial_states}). Namely, if $U=\textrm{C-NOT}$, then 
$\Phi_U\rho_1= {\rm diag}\left(\frac{2}{5}, \frac{3}{10}, \frac{1}{10}, \frac{1}{5} \right) \neq \rho_1$, $\Phi_U\rho_2 = \rho_2$ and $\Phi_U\rho_3= \rho_3$.  If $U=\textrm{C-PHASE} (\lambda)$, then
$\Phi_U\rho_1= \rho_1$, $\Phi_U\rho_3= \rho_3$ and
$$\Phi_U\rho_2=
\frac{1}{4}\begin{pmatrix}
1 & 1 & 1 & e^{-i\lambda} \\
1 & 1 & 1 & e^{-i\lambda} \\
1 & 1 & 1 & e^{-i\lambda} \\
e^{i\lambda} & e^{i\lambda} & e^{i\lambda} & 1
\end{pmatrix} \neq \rho_2.$$

\section{Gradient and Hessian of the Objective Functionals}\label{Sec:AppendixB}

Practical application of gradient and Hessian-based algorithms relies on analytical formulas for the gradient and the Hessian of the objective functionals, as in GRAPE~\cite{Khaneja_JMagnReson_2005} of BGFS~\cite{EitanPRA2011, deFouquieres2011}. Moreover, explicit formulas for the Hessian evaluated at critical points of the objective functional are necessary for the analysis of presence or absence of the trapping behaviour of objective functional~\cite{RabitzHsiehRosenthal2004,Pechen2011,FouquieresSchirmer,VolkovPechenUMN2023}.
In this Appendix, we derive the formulas for the gradient and the Hessian of the objective functionals~$F_U^{\rm sd}$, $F_U^{\rm GRK,sd}$, and $F_U^{\rm GRK,sp}$ for any $N$. Then we substitute piecewise constant controls into the formulas for the gradient and obtain expressions used in this work when applying inGRAPE in Appendix~\ref{Sec:AppendixC}.

For generality we provide general formulas for functional variation for a~system whose controlled dynamics is determined by a~linear evolutionary equation. Consider as a~control space a~normed space $\mathcal U$ continuously embedded into the space $L_1=L_1([0,T];\mathbb R^d)$.  For example, someone can choose $L_1$ itself, $L_2=L_2([0,T];\mathbb R^d)$,  $L_\infty=L_\infty([0,T];\mathbb R^d)$ or the space of piecewise constant $\mathbb R^d$-valued functions with $L_\infty$-norm  as $\mathcal U$. 
The Hilbert space $L_2$ is convenient  for studying the Hessian~\cite{FouquieresSchirmer,VolkovMorzhinPechenJPA2021}; the space $L_\infty$ is used when there are restrictions on the value of the control function.
Equation~(\ref{Eq:ME2_evolution_operator}) is a~particular case of the general evolutionary equation with control $f \in \mathcal U$:
\begin{equation}
\label{QCL_Phi_eq}
\dot{\Phi}_t^f = \mathcal{L}^{f}_t\Phi_t^f, \quad \Phi^f_0 = \mathbb{I},
\end{equation}
where $\mathcal{L}_t^f:=\mathcal{K}+f_\mu(t)\mathcal{N}^\mu$. Here $\mathcal{K}$, $\mathcal{N}^\mu$
($\mu=1,\ldots,d)$ are linear operators on  a~finite dimensional space $\mathcal{V}$.   In this section, we summarize by repeating Greek indices.
Carath\'eodory's theorem implies that for $f\in L_1$  equation~(\ref{QCL_Phi_eq}) has an unique absolutely continuous solution \cite{FilippovBook}. 
The solution of~(\ref{QCL_Phi_eq})  is determined through a~chronological exponent:
\begin{multline}
\Phi_t^f =\;\hat{T} \exp\!\left({\displaystyle\int_0^t \mathrm{d} s\, \mathcal{L}^{f}_s}\right)= \mathbb{I} + \displaystyle \sum_{n =  1}^{\infty} \dfrac{1}{n!} \int_0^t  \dots \int_0^{t} \hat{T} \left\{\mathcal{L}^{f}_{\tau_n}\dots \mathcal{L}^{f}_{\tau_1}\right\} \mathrm{d} \tau_1 \dots \mathrm{d}\tau_n = \\ 
=\mathbb{I} + \int_0^t   \mathcal{L}_s^f \,\mathrm{d}s + \frac 12  \int_0^t \mathrm{d}\tau \!\int_0^{\tau} \!\hat{T} \left\{\mathcal{L}^{f}_\tau\mathcal{L}^{f}_s\right\} \, \mathrm{d}s  + \dots,
\label{QCL_ordered_exponential}
\end{multline}
where $\hat{T}\{\,\cdot\,\}$ is the
chronological ordering operator, which  sets the multipliers in the chronological order of their application in the composition of operators. 
The series~(\ref{QCL_ordered_exponential}) converges for finite-dimensional operators.
About the notion of chronological exponent, also see~\cite[Ch.~2]{AgrachevBook2004}. 

Let $\mathcal{J}$ be a~ $C^2$-smooth function on  $GL(\mathcal{V})$, the general linear group on $\mathcal{V}$.
Our goal is to calculate the gradient and the  Hessian of an arbitrary functional of the form
\begin{equation}
\label{arbitraryfunctional}
\mathcal{F}(f) = \mathcal{J}(\Phi_T^f),
\end{equation}
that is, to calculate the first and second derivatives in the sense of Fr\'echet in the functional space of controls.

Let us introduce the notation  ${(\N_t^{f})}^\mu = {\Phi_t^f}^{-1}\!\mathcal{N}^\mu \Phi_t^f.$ This operator exists due to invertibility of the evolution operator
\begin{equation*}
{\Phi_t^f}^{-1} = \hat{T}_a \exp\!\left( - {\displaystyle\int_0^t  \mathcal{L}^{f}_s \,\mathrm{d}s}  \right),
\end{equation*}
where $\hat{T}_a$ is the antichronological ordering operator, i.e., unlike the $\hat{T}$ operator, this operator rearranges the factors in the reverse chronological  order of their appearance.

\begin{proposition}
\label{propositionMainGrHess}
The first and the second Fr\'echet derivatives of the functional~(\ref{arbitraryfunctional}) in the function space~$\mathcal U$ are
\begin{align}
\label{QCL_functional_gradient}
\dfrac{\delta\mathcal{F}(f)}{\delta f_\mu(t)}  =& \dfrac{\delta\mathcal{J}\left(\Phi_T^f\right)}{\delta\Phi}  \Phi_T^f{(\N_t^{f})}^\mu,\\
% \label{QCL_functional_Hessian}
\frac{\delta ^2 \mathcal{F}(f)}{\delta f_\mu(t_1)\,\delta f_\nu(t_2)}  =& \dfrac{\delta\mathcal{J}\left(\Phi_T^f\right)}{\delta\Phi}  \Phi_T^f \hat{T}\{{(\N^{f}_{t_1})}^\mu{(\N^{f}_{t_2})}^\nu\} + \dfrac{\delta^2\mathcal{J}\left(\Phi_T^f\right)}{\delta\Phi^2} \left(\Phi_T^f{(\N^{f}_{t_1})}^\mu, \Phi_T^f{(\N^{f}_{t_2})}^\nu \right).\nonumber
\end{align}
\end{proposition}
\begin{proof}
To calculate the derivatives of the objective functional, consider the increment $\delta f$ in the neighborhood of the control $f$.
Introduce the operator
$W^{f,\delta f}_t$,
where
\begin{equation*}
\Phi^{f + \delta f}_t = \Phi_t^f W^{f,\delta f}_t.
\end{equation*}
This operator exists due to invertibility of the evolution operator.
The operator $\Phi^{f+\delta f}_t$   satisfies equation~(\ref{QCL_Phi_eq}) with control $f+\delta f$. This allows us to obtain the following equation for $W^{f,\delta f}_t$:
\begin{equation*}
\dot{W}^{f,\delta f}_t= \delta f_\mu(t) \, {\Phi_t^f}^{-1}\!\mathcal{N}^\mu \Phi_t^f W^{f,\delta f}_t, \quad W^{f,\delta f}_0 = \mathbb{I}.
\end{equation*}
Its solution is 
\begin{equation*}
W^{f,\delta f}_T = \hat{T} \exp\!\left(\displaystyle\int_0^T  \delta f_\mu(t) {(\N_t^{f})}^\mu \,\mathrm{d}t \right).
\end{equation*}
Using~(\ref{QCL_ordered_exponential}), we obtain the Taylor expansion
\begin{multline}
\label{QCL_Phi_tailor}
\Phi^{f+\delta f}_T=\Phi_T^f W^{f,\delta f}_T = \Phi_T^f + \int_0^T \delta f_\mu(t) \Phi_T^f{(\N_t^{f})}^\mu \, \mathrm{d}t +\\ 
+  \dfrac{1}{2} \int_0^T\int_0^T  \Phi_T^f \hat{T}\{{(\N_{t_1}^{f})}^\mu{(\N_{t_2}^{f})}^\nu\} \,\delta f_\mu(t_1) \delta f_\nu(t_2) \, \mathrm{d}t_1\,\mathrm{d}t_2 +  \dots 
\end{multline}
Then the first Fr\'echet derivative of the control evolution operator equals to
\begin{equation*}
\dfrac{\delta\Phi_T^f}{\delta f_\mu(t)} = \Phi_T^f{(\N_t^{f})}^\mu.
\end{equation*}
Hence, the derivative of the functional~(\ref{arbitraryfunctional}) equals to
\begin{equation*}
\frac{\delta\mathcal{F}^f}{\delta f_\mu(t)}  = \frac{\delta\mathcal{J}\left(\Phi_T^f\right)}{\delta \Phi}  \frac{\delta\Phi_T^f}{\delta f_\mu(t)}= \frac{\delta\mathcal{J}\left(\Phi_T^f\right)}{\delta\Phi} \Phi_T^f{(\N_t^{f})}^\mu.
\end{equation*}

Moreover, decomposition~(\ref{QCL_Phi_tailor}) gives an expression for the second-order derivative of the evolution operator:
\begin{equation*}
\dfrac{\delta^2\Phi_T^f}{\delta f_\mu(t_1) \delta f_\nu(t_2)} = \Phi_T^f \hat{T}\{{(\N_{t_1}^{f})}^\mu{(\N_{t_2}^{f})}^\nu\}.
\end{equation*}
The formula for the second derivative of an arbitrary functional~(\ref{arbitraryfunctional}) is obtained by differentiating~(\ref{QCL_functional_gradient}) as follows
\begin{equation*}
\frac{\delta ^2 \mathcal{F}(f)}{\delta f_\mu(t_1) \delta f_\nu(t_2)} = 
\dfrac{\delta\mathcal{J}\left(\Phi_T^f\right)}{\delta\Phi}  \Phi_T^f \hat{T}\{{(\N_{t_1}^{f})}^\mu{(\N_{t_2}^{f})}^\nu\} + \dfrac{\delta^2\mathcal{J}\left(\Phi_T^f\right)}{\delta\Phi^2} \left(\Phi_T^f{(\N_{t_1}^{f})}^\mu, \Phi_T^f{(\N_{t_2}^{f})}^\nu\right),
\end{equation*}
where $\dfrac{\delta^2\mathcal{J}}{\delta\Phi^2} (\,\cdot\,,\,\cdot\,)$, being a~second-order derivative, is a~bilinear map.
\end{proof}

Let the control $f=(f_1,f_2,f_3)=(u,n_1,n_2)$ belong to the interior of the set of admissible controls $\{f\in L_\infty([0,T],\mathbb{R}^3)\colon f_2\geq 0,f_3\geq 0\}$.  Proposition~\ref{propositionMainGrHess} implies that the first  and the second Fr\'echet  derivatives of the objective functional $F_U$
at the control point $f$ have the form
\begin{align*}
(F_{U})'[f](\delta f)&=\int_0^T\frac{\delta F_{U}[f]}{\delta f_\mu(t)}\delta f_\mu(t)dt,\\
(F_{U})''[f](\delta f_1,\delta f_2)&= \int_0^T\int_0^T\frac{\delta^2{F_{U}}[f]}{\delta f_\mu(t_1) \delta f_\nu(t_2)} (\delta f_1)_\mu(t_1)(\delta f_2)_\nu(t_2)dt_1dt_2,
\end{align*}
where $F_U$ means $F_U^{\rm sd}$ or $F_U^{\rm GRK,sd}$ or $F_U^{\rm GRK,sp}$.
\begin{proposition}\label{propositionGradientHessianF1}
Gradient and Hessian of the functional $F_{U}^{\rm sd}$ have the forms
\begin{align*}
\frac{\delta F_{U}^{\rm sd}[f]}{\delta f_\mu(t)}&=\frac{1}{N^2}\mathrm{Tr}\left({(\Phi_T^f-\Phi_U)}^\dagger\Phi_T^f{(\mathcal{N}^{f}_t)}^\mu\right).\\
\frac{\delta^2{F_{U}^{\rm sd}[f]}}{\delta f_\mu(t_1) \delta f_\nu(t_2)}
&=\frac{1}{N^2}\mathrm{Tr}\left((\Phi_T^f-\Phi_U)^\dagger\Phi_T^f\hat{T}\{{(\N_{t_1}^{f})}^\mu{(\N_{t_2}^{f})}^\nu\}\right)+\frac{1}{N^2}\mathrm{Tr}\left((\Phi_T^f {(\N_{t_1}^{f})}^\mu)^\dagger \Phi_T^f {(\N_{t_2}^{f})}^\nu\right).
\end{align*}
\end{proposition}

\begin{proposition}\label{propositionGradientHessianF2}

Gradient and Hessian of the functional $F^{\rm GRK,sd}_{U}$ have the forms
\begin{align*}
\frac{\delta F^{\rm GRK,sd}_{U}[f]}{\delta f_\mu(t)}=&\sum_{m=1}^3\frac 13\mathrm{Tr}\left((\Phi_T^f\rho_{m}-\Phi_U\rho_{m})\Phi_T^f{(\mathcal{N}^{f}_t)}^\mu\rho_{m}\right).\\
\frac{\delta^2{F^{\rm GRK,sd}_{U}[f]}}{\delta f_\mu(t_1) \delta f_\nu(t_2)}=&\sum_{m=1}^3\frac 13\mathrm{Tr}\left((\Phi_T^f\rho_{m}-\Phi_U\rho_{m})\Phi_T^f\hat{T}\{{(\N_{t_1}^{f})}^\mu{(\N_{t_2}^{f})}^\nu\}\rho_{m}\right)\\
&+\sum_{m=1}^3\frac 13\mathrm{Tr}\left(\Phi_T^f {(\N_{t_1}^{f})}^\mu \rho_{m}\Phi_T^f {(\N_{t_2}^{f})}^\nu\rho_{m}\right).
\end{align*}
\end{proposition}
\begin{proposition}
 \label{propositionGradientHessianF3}

Gradient and Hessian of the functional $F_U^{\rm GRK,sp}$ have the forms
\begin{align*}
\frac{\delta F_U^{\rm GRK,sp}[f]}{\delta f_\mu(t)} &= 
-\sum_{m=1}^3\frac 1{3\mathrm{Tr} {\rho}^2_m} \left(\mathrm{Tr}\left[\Phi_U{\rho}_m\Phi_T^f{(\mathcal{N}^{f}_t)}^\mu\rho_m\right]\right).\\
\frac{\delta^2 F_U^{\rm GRK,sp}[f]}{\delta f_\mu(t_1) \delta f_\nu(t_2)}
&=-\sum_{m=1}^3 \frac 1{3\mathrm{Tr} {\rho}^2_m}\left(\mathrm{Tr}\left[\Phi_U{\rho}_m\Phi_T^f\hat{T}\{{(\N_{t_1}^{f})}^\mu{(\N_{t_2}^{f})}^\nu\}\rho_m\right]\right)
\end{align*}
\end{proposition}

\section{Realification of the Quantum System and of the Objective Functionals}  \label{Sec:AppendixC}

Realification of the set of density matrices and the set of operators on this set is important in the context of optimization in order to reduce computational space and time cost. It is done by performing only real-value calculations instead of complex-value calculations. As it will be shown further, this enables reduction by half the dimension of spaces carrying states and operators on states.

\subsection{Parametrization of the Density Matrix}

Taking into account the hermiticity of the density matrix, we consider a~parameterization of density matrix $\rho$ by a~real vector $x=x_\rho=(x_1,\ldots,x_{16}) \in \R^{16}$. More specifically, we consider the following expansion in the special Hermitian basis $\{M_k\}$:
\begin{equation}
\rho = \sum_{j = 1}^{16} x_j M_j 
= \begin{pmatrix}
x_1 & x_2 + i x_3 & x_4 + i x_5 & x_6 + i x_7 \\
x_2 - i x_3 & x_8 & x_9 + i x_{10} & x_{11} + i x_{12} \\
x_4 - i x_5 & x_9 - i x_{10} & x_{13} & x_{14} + i x_{15} \\
x_6 - i x_7 & x_{11} - i x_{12} & x_{14} - i x_{15} & x_{16} 
\label{rho_parameterization}
\end{pmatrix}.
\end{equation}
The condition ${\rm Tr}\rho=1$ implies linear dependence $x_1+x_8+x_{13}+x_{16}=1$.

Rewritten for the coordinate vector $x$, the dynamical systems 1, 2, and 3 have the following form:
\begin{equation}
\label{dynamical_system_x_common_form}
\der{x^{u,n}}{t} = L^{u,n}_t x, \qquad x^{u,n}(0) = x_{\rho_0},
\end{equation}
where $L_t^{u,n}$ is the matrix of the generator $\L_t^{u,n}$ in the basis $M = \{M_j\}_{j = 1}^{16}$~(\ref{rho_parameterization}):
\begin{equation*}
    L_t^{u,n} = \left(A + B_u u(t) + B_{n_1} n_1(t) + B_{n_2} n_2(t) \right),
\end{equation*}
the $16 \times 16$ matrices $A$, $B_u$, $B_{n_1}$, $B_{n_2}$ are found by substituting the expansion~(\ref{rho_parameterization}) into the master equation~(\ref{Eq:ME2}) with the Hamiltonian~(\ref{total_Hamiltonian}) and the dissipator~(\ref{dissipator}) for each system 1, 2, and 3 defined in Section~\ref{Sec:ME}; $x_{\rho_0}$ is the coordinate vector of $\rho_0$. For the interaction Hamiltonians $V_1$ and $V_2$, these dynamical systems and initial conditions were explicitly written in~\cite{Morzhin_Pechen_LJM_2021}. For brevity, here we do not reproduce these equations and the corresponding matrices. 

Introduce an evolution operator $\Psi^{u,n}_t$ which is a~matrix of the evolution operator $\Phi^{u,n}_t$ in the basis $M = \{M_j\}_{j = 1}^{16}$~(\ref{rho_parameterization}), gives evolution of a~vector~$x$:
$x^{u,n}(t) = \Psi^{u, n}_t x_{\rho_0}$, and satisfies the equation:%matrix version of the system~(\ref{dynamical_system_x_common_form}):
\begin{equation}
\der{\Psi}{t} =  L_t^{u,n} \Psi, \qquad \Psi_0 = \I. \label{evolution_matrix}
\end{equation}
Note that any quantum channel $\Phi$ maps density matrices to density matrices, therefore a~matrix $\Psi$ of any quantum channel is real-valued in the Hermitian basis:
\begin{equation*}
    \Psi \in \R^{N\times N} \subset \C^{N\times N}.
    \label{Psi_real}
\end{equation*}

\subsection{Objective Functionals and Their Gradients in Terms of Real-Valued States}

Here we provide expressions for the objective kinematic functionals and their gradient in the real-valued parameterization proposed in~(\ref{rho_parameterization}). Firstly, we rewrite expressions for the objective kinematic functionals $J_U^{\rm sd}$, $J_U^{\rm GRK,sd}$, and $J_U^{\rm GRK,sp}$ as functionals of a~matrix $\Psi$ of an operator $\Phi$ in the basis $M = \{M_j\}_{j = 1}^{16}$. Then we rewrite expressions for the gradient of the objective dynamic functionals in a convenient and effective for optimization form.

The first functional $J_U^{\rm sd}$ needs calculation of the squared Hilbert--Schmidt norm $\|\Phi\|^2 = \Tr(\Phi^\dagger\Phi)$. For that we prove the following proposition.
\begin{proposition}
    Let $\Psi$ and $\Psi'$ be matrices of operators $\Phi$ and $\Phi'$ from $\L(\C^{N\times N})$, respectively, in an orthogonal basis $M = \{M_j\}_{j = 1}^{N^2}$, $M_j \in \C^{N\times N}$, $\scalarproduct{M_i}{M_j} = \Tr (M_i^\dagger M_j) = \beta_j\delta_{ij}$; $i,j = 1, \ldots, N^2$. Then the Hilbert--Schmidt scalar product of $\Phi, \Phi' \in \L(\C^{N\times N})$ equals
    \begin{equation}
        \scalarproduct{\Phi}{\Phi'} = \sum_{i,j = 1}^{N^2}\frac{\beta_i}{\beta_j}\overline{\Psi}_{ij}\Psi'_{ij}.
        \label{proposition_HS_sp}
    \end{equation}
\end{proposition}

\begin{proof}
    Let $\rho$ and $\sigma$ be matrices in $\C^{N\times N}$ and have coordinates $x$ and $y$, respectively, in the orthogonal basis $M$. Then their scalar product equals
    \begin{equation}
        \scalarproduct{\rho}{\sigma} = \overline{(\beta \circ x)}^Ty,
        \label{scalar_product_parameterization}
    \end{equation}
    where ``$\circ$`` denotes the Hadamard product that returns the vector which components equal
    \begin{equation*}
        (\beta \circ x)_j = \beta_j x_j,\quad j = 1, \ldots, N^2.
    \end{equation*}
    Denote by ${\Psi}^\dagger$ a~matrix of ${\Phi}^\dagger$ in the basis $M$. From the equality $\scalarproduct{\rho}{\Phi \sigma} = \scalarproduct{{\Phi}^\dagger \rho}{\sigma}$ we have that the components of ${\Psi}^\dagger$ equal
    \begin{equation*}
        {\Psi}^\dagger_{ij} = \frac{\beta_i}{\beta_j}\overline{\Psi}_{ji}.
    \end{equation*}
    Then the Hilbert--Schmidt scalar product equals $$\scalarproduct{\Phi}{\Phi'} = \Tr {\Psi}^\dagger\Psi' = \sum_{i,j = 1}^N\frac{\beta_i}{\beta_j}\overline{\Psi}_{ij}\Psi'_{ij}.$$ 
\end{proof}

The considered parameterization basis $M = \{M_j\}_{j = 1}^{16}$~(\ref{rho_parameterization}) is orthogonal, the vector of squared norms $\beta_j = \Tr M_j^\dagger M_j$ equals
\begin{equation*}
    \beta = (1, 2, 2, 2, 2, 2, 2, 1, 2, 2, 2, 2, 1, 2, 2, 1).
\end{equation*}
Equations~(\ref{proposition_HS_sp}) and~(\ref{scalar_product_parameterization}) justify introducing scalar products and norms in $\R^{16}$ and in $\R^{16\times16}$:
\begin{align*}
    \scalarproduct{x}{y}_M &= \scalarproduct{\beta \circ x}{y} = \sum_{j = 1}^{16} \beta_j x_j y_j,\quad \norm{x}^2_M = \scalarproduct{x}{x}_M,\quad x,y \in \R^{16};\\
    \scalarproduct{\Psi}{\Psi'}_M &= \sum_{i,j = 1}^{16}\frac{\beta_i}{\beta_j}\Psi^1_{ij}\Psi'_{ij},\quad \norm{\Psi}^2_M = \scalarproduct{\Psi}{\Psi}_M,\quad \Psi, \Psi'\in \R^{16\times16}.
\end{align*}
Thus, we can write the objective kinematic functionals as a~function of $\Psi$ as follows:
\begin{proposition}
The objective kinematic functionals $J_U^{\rm sd}$, $J_U^{\rm GRK,sd}$, and $J_U^{\rm GRK,sp}$ as functionals of a~matrix $\Psi$ of an operator $\Phi$ in the basis $\{M_j\}_{j = 1}^{16}$ used in~(\ref{rho_parameterization}) are equal to
\begin{align*}
J_U^{\rm sd}(\Phi) &= \frac1{32}\norm{\Psi - \Psi_U}^2_M, \\
J_U^{\rm GRK,sd}(\Phi) &= \frac{1}{6}\sum_{m = 1}^3 \norm{\Psi x_{\rho_m} - \Psi_Ux_{\rho_m}}^2_M, \\
 J_U^{\rm GRK,sp}(\Phi) &= 1 - \frac{1}{3}\sum_{m = 1}^3 \dfrac{\scalarproduct{\Psi x_{\rho_m}}{\Psi_U x_{\rho_m}}_M}{\norm{x_{\rho_m}}^2_M}.
\end{align*}
where  $\Psi_U$ is a~matrix of $\Phi_U$ and $x_{\rho_m}$ is the coordinate vectors of $\rho_m$, $m = 1,2,3$, in the parameterization basis~(\ref{rho_parameterization}).
\end{proposition}

Let control $f = (u, n_1, n_2)$ be a piecewise constant control given by~(\ref{piecewise_constant_controls}). In addition,  consider control $g = (u, w_1, w_2)$, where $w_1$ and $w_2$ are given by the change of variable~(\ref{piecewise_constant_control_w}). Denote by $x_{\rho_m}^f(t)$ the solution of the system~(\ref{dynamical_system_x_common_form}) with the initial condition $x_{\rho_m}^f(0) = x_{\rho_m}$. When considering piecewise constant controls, the generator matrix $L^f_t$ is constant on the intervals: $L_t^f = L^f_{t_{k-1}} = A + B_uu_k + B_{n_1}n_{k,1} + B_{n_2}n_{k,2}$ for any $t \in [t_{k-1}, t_k)$. The evolution matrix values~$\Psi_{t_k}^f$ on the borders $t_k$ are equal to $\Psi^f_{t_k} = \exp(\Delta t L^f_{t_{k-1}}) \cdots \exp(\Delta t L^f_0)$. For the considered piecewise constant control, the objective functional $F_U$ become function of $3K$ variables $g_{k,\mu}$. 

For the optimization process described in  Section~\ref{Section5}  we use the following expressions for gradient of the objective functionals $F_U^{\rm sd}$, $F_U^{\rm GRK,sd}$, and $F_U^{\rm GRK,sp}$ which are derived as direct corollary of Propositions~\ref{propositionGradientHessianF1},~\ref{propositionGradientHessianF2},~and~\ref{propositionGradientHessianF3} in Appendix~\ref{Sec:AppendixB}.
\begin{proposition}
    Gradient of the functional $F_U^{\rm sd}$ has the form
    \begin{equation}
        \pd{F_U^{\rm sd}}{g_{k,\mu}} = \frac1{16}\scalarproductleftright{\Psi_T^f - \Psi_U}{\pd{\Psi^f_T}{g_{k,\mu}}}_M.\label{gradient_pconst_F1}
    \end{equation}
\end{proposition}
\begin{proposition}
    Gradient of the functional $F_U^{\rm GRK,sd}$ has the form
   \begin{equation}
        \pd{F_U^{\rm GRK,sd}}{g_{k,\mu}}  = \frac13\sum_{m = 1}^3\scalarproductleftright{\Psi_T^f x_{\rho_m} - \Psi_Ux_{\rho_m}}{\pd{x^f_{\rho_m}(T)}{g_{k,\mu}}}_M.
    \end{equation}\label{gradient_pconst_F2}
\end{proposition}
\begin{proposition}
    Gradient of the functional $F_U^{\rm GRK,sp}$ has the form
    \begin{equation}
        \pd{F_U^{\rm GRK,sp}}{g_{k,\mu}}  = -\frac13\sum_{m = 1}^3\frac{1}{\norm{x_{\rho_m}}^2_M}\scalarproductleftright{\Psi_U x_{\rho_m}}{\pd{x^f_{\rho_m}(T)}{g_{k,\mu}}}_M.
    \end{equation}\label{gradient_pconst_F3}
\end{proposition}

\begin{proposition}
    Gradient of the evolution matrix $\Psi_T^f$ and the final states $x_{\rho_m}(T,f)$ equals
    \begin{equation}
    \begin{split}
        \pd{\Psi_T^f}{u_k} &= \Delta t\Psi^f_{t_k, T}\int_0^1\exp\left((1-\tau)\Delta t L^f_{t_{k-1}}\right)B_u\exp\left(\tau\Delta t L^f_{t_{k-1}}\right)\d \tau \Psi_{t_{k-1}}^f,\\
        \pd{\Psi_T^f}{w_{k,l}} &= 2w_{k,l}\Delta t\Psi^f_{t_k, T}\int_0^1\exp\left((1-\tau)\Delta t L^f_{t_{k-1}}\right)B_{n_l}\exp\left(\tau\Delta t L^f_{t_{k-1}}\right)\d \tau \Psi_{t_{k-1}}^f\\
        \pd{x^f_{\rho_m}(T)}{u_k} &= \Delta t\Psi^f_{t_k, T}\int_0^1\exp\left((1-\tau)\Delta t L^f_{t_{k-1}}\right)B_u\exp\left(\tau\Delta t L^f_{t_{k-1}}\right)\d \tau x^f_{\rho_m}(t_{k-1}),\\
        \pd{x^f_{\rho_m}(T)}{w_{k,l}} &= 2w_{k,l}\Delta t\Psi^f_{t_k, T}\int_0^1\exp\left((1-\tau)\Delta t L^f_{t_{k-1}}\right)B_{n_l}\exp\left(\tau\Delta t L^{f}_{t_{k-1}}\right)\d \tau x^f_{\rho_m}(t_{k-1}),\label{gradient_pc_Psi_x}
    \end{split}
    \end{equation}
    $k = 1, \ldots, K$; $l = 1,2$; $m = 1, 2, 3$.
\end{proposition}
Here the evolution matrix from time $s$ to time $t$ is defined as $\Psi_{s,t}^f = \Psi_t^f{\Psi_s^f}^{-1}$. Its values for $s = t_k$ and $t = T$ are equal to $\Psi^f_{t_k,T} = \exp(\Delta t L^f_{t_{K-1}}) \cdots \exp(\Delta t L^f_{t_{k}}).$

\vspace{0.3cm}

\noindent {\bf Abbreviations:} 
\begin{itemize}
\item GKSL -- Gorini--Kossakowski--Sudarshan--Lindblad;
\vspace{-0.3cm}
\item C-NOT -- controlled NOT (quantum gate);
\vspace{-0.3cm}
\item C-PHASE -- Controlled PHASE (quantum gate);
\vspace{-0.3cm}
\item C-Z -- Controlled Z (quantum gate);
\vspace{-0.3cm}
\item GRK approach -- the approach developed by M.Y.~Goerz, D.M.~Reich, and C.P.~Koch in~\cite{Goerz_NJP_2014_2021, Goerz_2021} for generating unitary gates under dissipative evolution, where only the three initial density matrices are used in contrast to the complete basis; 
\vspace{-0.3cm}
\item DAA -- Dual Annealing Algorithm;
\vspace{-0.3cm}
\item GRAPE -- GRadient Ascent Pulse Engineering;
\vspace{-0.3cm}
\item inGRAPE -- incoherent GRAPE.
\end{itemize}

\end{document}